\documentclass[11pt]{article}
\usepackage{amsmath,amsthm,amssymb}
\usepackage{url}
\usepackage{color}
\usepackage{enumitem}
\usepackage{algpseudocode}
\usepackage{algorithm}
\usepackage{graphicx}
\usepackage[margin=1in]{geometry}

\usepackage{times}

\DeclareMathOperator{\supp}{supp}
\DeclareMathOperator*{\E}{\mathbb{E}}
\let\Pr\relax
\DeclareMathOperator*{\Pr}{\mathbb{P}}
\DeclareMathOperator*{\median}{median}
\DeclareMathOperator*{\quant}{quant}
\DeclareMathOperator{\poly}{poly}
\DeclareMathAlphabet{\pazocal}{OMS}{zplm}{m}{n}

\newtheorem{theorem}{Theorem}
\newtheorem{remark}[theorem]{Remark}
\newtheorem{lemma}[theorem]{Lemma}
\newtheorem{corollary}[theorem]{Corollary}
\newtheorem{proposition}[theorem]{Proposition}	
\newtheorem{fact}[theorem]{Fact}	
\theoremstyle{definition}
\newtheorem{definition}{Definition}

\newcommand{\Oh}{\mathcal{O}}

\newcommand{\eps}{\varepsilon}

\newcommand{\R}{\mathbb{R}}

\title{ On Low-Risk Heavy Hitters and Sparse Recovery Schemes 
}

\author{
	Yi Li\\
	Nanyang Technological University\\
	\texttt{yili@ntu.edu.sg}\\
\and
	Vasileios Nakos\footnote{Supported in part by NSF grant IIS-144741.}\\
	Harvard University \\
	\texttt{vasileiosnakos@g.harvard.edu}
\and
	David P. Woodruff\footnote{Supported in part by NSF grant CCF-1815840.}\\
	Carnegie Mellon University\\
	\texttt{dwoodruf@cs.cmu.edu}
}
\date{}

\allowdisplaybreaks

\begin{document}

\maketitle

\begin{abstract}
We study the heavy hitters and related sparse recovery problems 
in the {\it low-failure probability regime}. 
This regime is not well-understood, 
and the main previous work on this is by Gilbert et al.\@ (ICALP'13). 
We recognize an error in their analysis, improve their results, and contribute new sparse recovery algorithms, as well as provide upper and lower bounds for the heavy hitters problem with low failure probability. Our results are
summarized as follows: 

\begin{enumerate}
\item (Heavy Hitters) We study three natural variants for finding heavy hitters in the strict turnstile model, where the variant depends on the quality of the desired output. 
For the weakest variant, we give a randomized algorithm improving the failure probability analysis of the ubiquitous \textsc{Count-Min} data structure. We also give a new lower bound for deterministic schemes, resolving a question about this variant posed in Question 4 in the IITK Workshop on Algorithms for Data Streams (2006). Under the strongest and well-studied
  $\ell_{\infty}/ \ell_2$ variant, 
  we show that the classical \textsc{Count-Sketch} data structure is optimal 
for very low failure probabilities, which was previously
unknown. 
\item (Sparse Recovery Algorithms)  For non-adaptive sparse-recovery, we give sublinear-time algorithms with low-failure probability, which improve upon Gilbert et al.\@ (ICALP'13). In the adaptive case, we improve the failure probability from a constant by Indyk et al.\@ (FOCS '11) to $e^{-k^{0.99}}$, where $k$ is the sparsity parameter.

\item (Optimal Average-Case Sparse Recovery Bounds) We give matching upper and lower bounds in all parameters, including the failure probability, for the measurement complexity of the $\ell_2/\ell_2$ sparse recovery problem in the spiked-covariance model, completely settling its complexity in this model. 
\end{enumerate}

\end{abstract}

\section{Introduction}
Finding heavy hitters in data streams is one of the most practically and theoretically important problems in the streaming literature. Subroutines for
heavy hitter problems, and in particular for the \textsc{Count-Min} sketch, 
are the basis for multiple problems on geometric
data streams, including $k$-means and $k$-median clustering \cite{i04,fs05,BFLSY17}, as well as image acquisition \cite{gipr12}. 
Studying schemes for finding such heavy hitters has also led to 
important geometric insights in $\ell_p$-spaces \cite{zgr16}. 

Abstractly, in the heavy hitters problem, we are asked to report all frequent items in a very long stream of elements coming from some universe. The main restriction is that the memory consumption should be much smaller than the universe size and the length of the stream. To rephrase the problem, consider a frequency vector $x\in \R^n$, where $n$ is the size of the universe. Each element $i$ in the data stream updates the frequency vector as $x_i\gets x_i+1$. At the end of the data stream, we wish to find the coordinates of $x$ for which $|x_i|\geq \epsilon \|x\|_1$.

Oftentimes the problem is considered under a more general streaming model called the strict turnstile model, where arbitrary deletions and additions are allowed, but at all times the entries of $x$ remain non-negative, that is $x_i \geq 0$. More formally, the frequency vector $x \in \mathbb{R}^n$ receives updates of the form $(i,\Delta)$, and each such update causes the change of $x_i$ to $x_i + \Delta$, while ensuring that $x_i \geq 0$. The goal is to maintain a data structure such that upon query, the data structure returns the heavy hitters of the underlying vector $x$. The $\ell_p$ heavy hitters problem, for $p \geq 1$, then asks to find all coordinates $i$ for which $x_i^p \geq \epsilon \|x\|_p^p$. The algorithm that treats the $\ell_1$ case is the Count-min sketch \cite{cormode2005improved}, and the algorithm that treats the $\ell_2$ case is the \textsc{Count-Sketch} \cite{CCFC}. Both algorithms are randomized and succeed with probability $1-1/\poly(n)$. In \cite{cormode2005improved} the authors also suggest the ``dyadic'' trick for exchanging query time with space.
Their ``dyadic'' trick allows for finding heavy hitters approximately in $\Oh(\frac{1}{\epsilon} \log^2 n)$ time, but with the downside of having a worse update time and a worse space consumption by an $\Oh( \log n)$ factor. The state of the art for heavy hitters is \cite{larsen2016heavy}, where the authors give an algorithm that satisfies the $\ell_{\infty}/l_p$ guarantee, has space $\Oh(\frac{1}{\epsilon} \log n)$, update time $\Oh(\log n)$, and query time $\Oh(\frac{1}{\epsilon} \poly(\log n))$. We note that the latter algorithm works in the more general setting of the turnstile model, where there is no constraint on $x_i$, in contrast to the strict turnstile model.

Another set of closely related problems occurs in the compressed sensing (CS) literature, which has seen broad applications to biomedical imaging, sensor networks, hand-held digital cameras, and monitoring systems, among other places. Sparse compressed sensing schemes were used for determining the attitudes\footnote{See \url{https://en.wikipedia.org/wiki/Attitude_control} for the notion of `attitude' in this context.}, or $3$-axis orientation, of spacecraft in \cite{gipr12}.  

Abstractly, in this problem we also seek to find the large coordinates of $x\in\R^n$ but with a different goal. Instead of finding all coordinates $x_i$ for which $|x_i|\geq \epsilon\|x\|_1$, the CS problem seeks an approximation $\hat x$ to $x$ such that the difference vector $x-\hat x$ has small norm. In particular, we consider the $\ell_2/\ell_2$ error metric, that is, we require that $\|x-\hat x\|_2^2\leq (1+\epsilon)\|x_{-k}\|_2^2$, where $x_{-k}$ is the vector $x$ with the $k$ largest entries (in magnitude) removed. If all $\ell_2$ heavy hitters are found, it is clear that the norm of $x-\hat x$ can be made small, but the CS problem allows a small number of heavy hitters to be missed if their contribution to the approximation error $x-\hat x$ is small.

Gilbert et al.\ proposed the first sublinear-time algorithm for the $\ell_2/\ell_2$ problem that achieves $O(\frac{k}{\epsilon}\log\frac{n}{k})$ measurements with constant failure probability~\cite{GLPS12}. Earlier sublinear-time algorithms all contain several additional $\log n$ factors in their number of measurements. The optimality of $O(\frac{k}{\epsilon}\log\frac{n}{k})$ measurements was shown by Price and Woodruff~\cite{PW11}.  Later Gilbert et al.\ improved the failure probability to $n^{-k/\poly(\log k)}$~\cite{gilbert2013l2}, while their number of measurements has a poor dependence on $\epsilon$, which is at least $\epsilon^{-11}$. 

Despite the above works, our understanding of the complexity of heavy hitter
and compressed sensing schemes on the error probability is very limited. The question regarding failure probability of these schemes is a natural one for two reasons: first, it is strongly connected with the existence of uniform schemes via the probabilistic method, and, second, being able to amplify the failure probability of an algorithm in a non-trivial way without making parallel repetitions of it, makes the algorithm much more powerful application-wise. For sparse recovery schemes, our goal is to obtain the same measurements but with smaller failure probability, something we find important both from a practical and theoretical perspective- obtaining the correct number of measurements in terms of all parameters $\epsilon,k,n,\delta$ would be the end of the story for compressed sensing tasks, and a challenging quest; we note that previous algorithms achieved optimality with respect to $\epsilon,k,n$ only. From the practical side, a small enough failure probability would allow to re-use the same measurements all the time, since an application of the union-bound would suffice for all vectors that might appear in a lifetime; thus, application-wise, we could treat such a scheme as uniform. We start with formal definitions of the problems and then state in detail 
our improvements over previous work.

\subsection{Problem Formulation}\label{sec:problem_formulation}

For $x \in \mathbb{R}^n$, we define $H_k(x)$ to be the set of its largest $k$ coordinates in magnitude, breaking ties arbitrarily. For a set $S$ let $x_S$ be the vector obtained from $x$ by zeroing out every coordinate $i \notin S$. We also define $x_{-k} = x_{[n] \setminus H_k(x)}$. For the $\ell_2/\ell_2$-sparse recovery results we define $H_{k,\epsilon}(x) = \{ i \in [n] : |x_i|^2 \geq \frac{\epsilon}{k} \|x_{-k}\|_2^2\}$. 

Two common models in the literature are the strict turnstile model and the (general) turnstile model.

\textbf{Strict Turnstile Model}: Both insertions and deletions are allowed, and it is guaranteed that at all times $x_i \geq 0$.

\textbf{(General) Turnstile Model}: Both insertions and deletions are allowed, but there is no guarantee about the value of $x_i$ at any point in time.


%

A sketch $f : \mathbb{R}^m \rightarrow \mathbb{R}^n$ is a function that maps an $n$-dimensional vector to $m$ dimensions. In this paper, all sketches will be linear, meaning $f(x)=Ax$ for some $A \in \mathbb{R}^{m \times n}$. The sketch length $m$ will be referred to as the space of our algorithms. 

\begin{definition}[Heavy hitters]\label{def:heavy_hitter}
For $x\in \R^n$ and $p\geq 1$, a coordinate $x_i$ is called an $\epsilon$-heavy hitter in $\ell_p$ norm if $|x_i|^p\geq \epsilon\|x\|_p^p$. We consider the following three variants of the heavy hitters problem with different guarantees:
\begin{enumerate}[topsep=0pt,itemsep=-1ex,partopsep=1ex,parsep=1ex]
\item Return a list containing all $\epsilon$-heavy hitters but no $\epsilon/2$-heavy hitters.

\item Return a list $L$ of size $\Oh(1/\epsilon)$ containing all $\epsilon$-heavy hitters along with estimates $\hat{x}_i$ such that $|x_i - \hat{x}_i|^p \leq \epsilon \|x_{-\lceil 1/\epsilon \rceil}\|_p^p$ for all $i\in L$. 
\item Return a list of size $\Oh(1/\epsilon)$ containing all $\epsilon$-heavy hitters.
\end{enumerate}
When the algorithm is randomized, it has a parameter $\delta$ of failure probability; that is, the algorithm succeeds with probability at least $1-\delta$.
\end{definition}

The variant with Guarantee 2 above is also referred to as the $\ell_{\infty}/\ell_p$ problem. In this paper we focus on $p = 1$ and $p = 2$,
with corresponding 
\textsc{Count-Min} \cite{cormode2005improved} and
\textsc{Count-Sketch} \cite{CCFC} data structures that have been studied extensively.

We note that the strongest guarantee is guarantee $2$. It folkolore that Guarantee $2$ implies both Guarantee $1,3$, and Guarantee $1$ clearly implies Guarantee $3$. In applications, such as sparse recovery tasks, it is often the case that one does not need the full power of Guarantees $1,2$, but rather is satisfied with Guarantee $3$. The natural question that arises is whether one can gain some significant advantage under this Guarantee. Indeed, we show that Guarantee $3$ allows the existence of a uniform scheme, i.e. one that works for all vectors, in the strict turnstile model with the same space, in contrast to the other two Guarantees.

\begin{definition}[$\ell_2/\ell_2$ sparse recovery]
An $\ell_2/\ell_2$-recovery system $\mathcal{A}$ consists of a distribution $\mathcal{D}$ on $\mathbb{R}^{m \times n}$ and a recovery algorithm $\mathcal{R}$. We will write $\mathcal{A} = (\mathcal{D}, \mathcal{R})$. We say that $\mathcal{A}$ satisfies the $\ell_2/\ell_2$ guarantee with parameters $(n,k,\epsilon,m,\delta)$ if for a signal $x \in \mathbb{R}^n$, the recovery algorithm outputs $\hat{x} = \mathcal{R}(\Phi, \Phi x)$ satisfying 
\[
\mathbb{P}_{\Phi \sim \mathcal{D}}\left\{ \|\hat{x} - x \|_2^2 \leq (1+ \epsilon) \|x_{-k}\|_2^2 \right\} \geq 1 - \delta.
\]
\end{definition}
In the above definition, each coordinate of $\Phi x$ is called a \emph{measurement} and the vector $\Phi x$ is referred to as the measurement vector or just as 
the measurements. The probability parameter $\delta$ is referred to as the failure probability.

\subsection{Our Results}
\noindent\textbf{Heavy hitters.} Our first result is an improved analysis of the \textsc{Count-Min} sketch \cite{cormode2005improved} for Guarantee 3 under the change of the hash functions from $2$-wise to $\Oh(\frac{1}{\epsilon})$-wise independence. Previous analyses for Guarantees 1 and 2 use $\Oh(\frac{1}{\epsilon} \log \frac{n}{\delta})$ space; in contrast our analysis shows that this version of the \textsc{Count-Min} sketch satisfies Guarantee 3 with only $\Oh(\frac{1}{\epsilon} \log (\epsilon n ) + \log (\frac{1}{\delta}))$ space. Notably, the $\frac{1}{\epsilon}$ factor does not multiply the $\log(\frac{1}{\delta})$ factor. This result has two important consequences. First, it gives a uniform scheme for Guarantee 3; second, it implies an improved analysis of the classic dyadic trick \cite{cormode2005improved} for Guarantee 3 using $\Oh( \frac{1}{\epsilon} \log( \epsilon n) + \log n \log(\frac{ \log \epsilon n}{\delta})) $ space. For constant $\delta$, previous analyses of the dyadic trick needed space $\Oh( \frac{1}{\epsilon} \log n \log(\frac{\log n}{\epsilon}))$ but our analysis shows that $\Oh( \frac{1}{\epsilon} \log (\epsilon n) + \log (\epsilon n) \log \log (\epsilon n))$ space suffices. These results are summarized in Table~\ref{tab:hh}.

Regarding the lower bound, we give the first bound for Guarantee 2 with $p = 2$, which is simultaneously optimal in terms of $n,$ any $\epsilon > \frac{1}{n^{.99}}$, and the failure probability $\delta$. That is, we prove that the space has to be $\Omega(\frac{1}{\epsilon}\log \frac{\epsilon n}{\delta})$, which matches the upper bound of \textsc{Count-Sketch} \cite{CCFC} whenever $\epsilon > \frac{1}{n^{.99}}$. A lower bound of $\Omega(\frac{1}{\epsilon} \log (\epsilon n))$ was given in \cite{jowhari2011tight} and is valid for the full range of parameters of $\epsilon$ and $n$, but 
previous analyses cannot be adapted to obtain 
non-trivial lower bounds for $\delta < \frac{1}{n}$. Indeed, the lower bound instances
used in arguments in previous work 
have deterministic upper bounds using $O (\frac{1}{\epsilon} \log n)$ space. 

We also show a new randomized lower bound of $\Omega(\frac{1}{\epsilon}(\log n + \sqrt{\log\frac{1}{\delta}}))$ space for $p = 1$, provided that $\frac{1}{\epsilon} > \sqrt{\log\frac{1}{\delta}}$. Although not necessarily optimal, this lower bound is the first 
to show that a term involving
$\log \frac{1}{\delta}$  must multiply the $\frac{1}{\epsilon}$ factor for $p = 1$. The assumption that $\frac{1}{\epsilon} > \sqrt{\log\frac{1}{\delta}}$ is necessary, as there exist deterministic $O(\frac{1}{\epsilon^2})$ space algorithms for $p = 1$ \cite{gm06,nnw12}. For deterministic algorithms satisfying Guarantee 3 with $p=1$, we also show a lower bound of $\Omega(\frac{1}{\epsilon^2})$ measurements, which resolves Question 4 in the IITK Workshop on Algorithms for Data Streams \cite{gregor2007open}.


\medskip 
\noindent\textbf{Sparse Recovery.} We summarize previous algorithms in Table~\ref{tbl:sparse}. 

We give algorithms for the $\ell_2/\ell_2$ problem with failure probability much less 
than $1/\poly(n)$ whenever $k=\Omega(\log n)$. We present two novel algorithms, one running in $\Oh(k\poly(\log n))$ time and the other in $O(k^2\poly(\log n))$ time with a trade-off in failure probability. Namely, the first algorithm has a larger failure probability than the second one. The algorithms follow a similar overall framework to each other but are instantiated with different parameters. We also show how to modify the algorithm of \cite{gilbert2013l2} to obtain an optimal dependence on $\epsilon$, achieving a smaller failure probability along the way. All of these results are included in Table~\ref{tbl:sparse}. Our algorithms, while constituting a significant improvement over previous work, are still not entirely optimal. We show, however, that at least in the spiked covariance model, which is a standard average-case model of input signals, we can obtain optimal upper and lower bounds in terms of the measurement complexity. Combined with the identification scheme from \cite{gilbert2013l2} we also obtain a scheme with decoding time nearly linear in $k$, assuming that $k = n^{\Omega(1)}$.

Besides the above non-adaptive schemes, we also make contributions, in terms of the failure probability, for adaptive schemes. For adaptive sparse recovery, Indyk et al.\@ gave an algorithm under the $\ell_2/\ell_2$ guarantee~\cite{indyk2011power} using $\Oh((k/\epsilon) \log(\epsilon n/k))$ measurements and achieving constant failure probability. In followup work \cite{nswz18} adaptive schemes were designed for other $\ell_p/\ell_p$ error guarantees and improved bounds on the number of rounds were given; here our focus, as with the non-adaptive schemes we study, is on the error probability. We give a scheme that achieves failure probability $e^{-k^{1-\gamma}}$ for any constant $\gamma$, using the same number of measurements. Moreover, we present an algorithm for the regime when $k/\epsilon \leq \poly(\log n)$. Our scheme achieves the stronger $\ell_{\infty}/ \ell_2$ guarantee and fails with probability $1/\poly(\log n)$. Thus, our algorithms improve upon \cite{indyk2011power} in both regimes: in the high-sparsity regime we get an almost exponential improvement in $k$, and in the low-sparsity regime we get $1/\poly(\log n)$.

\begin{table}
\centering
\begin{tabular}{|c|c|c|c|}
\hline
Algorithm & Space & Guarantee & Query time\\
\hline
\textsc{Count-Min} \cite{cormode2005improved} & $\frac{1}{\epsilon} \log n + \frac{1}{\epsilon} \log(\frac{1}{\delta})$ & 1, 2 & $\tilde{\Oh}( n  )$ \\
\hline
This paper & $\frac{1}{\epsilon} \log (\epsilon n) + \log(\frac{1}{\delta})$ & 3 & $\tilde{\Oh}(n)$\\
\hline
Dyadic Trick \cite{cormode2005improved}& $\frac{1}{\epsilon} \log n  \log(\frac{\log n}{\delta \cdot \epsilon}) $ & 1, 2 & $\tilde{\Oh}(\frac{1}{\epsilon})$\\
\hline
This paper & $\frac{1}{\epsilon} \log( \epsilon n) + \log (\epsilon n) \log(\frac{ \log (\epsilon n)}{\delta}) $ & 3 & $\tilde{\Oh}(\frac{1}{\epsilon})$\\
\hline
\end{tabular}
\caption{Summary of previous heavy hitter algorithms. The notation $\Oh(\cdot)$ for space complexity is suppressed, $\tilde{\Oh}(\cdot)$ hides logarithmic factors in $n$, $1/\epsilon$ and $1/\delta$.}\label{tab:hh}
\end{table}

\begin{table}
\begin{center}
 \begin{tabular}{|c|c|c|c|} 
 \hline
 Paper & Measurements & Decoding Time  & 
 Failure Probability \\ 
 \hline
 \cite{CCFC} & $\epsilon^{-1} k \log n $  & $n\log n$ & 
 $1/\poly(n)$ \\
 \hline
 \cite{GLPS12} & $\epsilon^{-1} k \log( n / k  ) $  & $ \epsilon^{-1} k \poly( \log (n/k))$& 
 $\Omega(1)$  \\ 
 \hline
  \cite{gilbert2013l2} & $\epsilon^{-11} k \log(n/k) $ & $k^{2} \cdot \poly(\epsilon^{-1}\log n)$ & 
   $(n/k)^{-k / \log^{13} k}$ \\
  \hline
 \cite{larsen2016heavy} & $\epsilon^{-1} k \log n$ & $ \epsilon^{-1} k\cdot \poly(\log n) $ &
 $1/ \poly(n)$  \\
 \hline
 \hline
 This paper & $\epsilon^{-1} k \log (n/k)$  & $\epsilon^{-1}k^{1+\alpha} \log^3 n$ & 
 $ e^{-\sqrt{k}/ \log^3k}$  \\
            & $\epsilon^{-1} k \log (n/k)$  & $\epsilon^{-1}k^2(\log k)\log^{2+\gamma}(n/k)$ & 
            $e^{-k/\log^3 k}$  \\
& $\epsilon^{-1} k \log(\frac{n}{\epsilon k})$ & $\epsilon^{-1}k^2 \poly(\log n)$ &
$(n/k)^{-k/\log k}$\\
 \hline
\end{tabular}
\end{center}
\caption{Summary of previous sparse recovery results and the results obtained in this paper. The notation $\Oh(\cdot)$ is suppressed. The paper~\cite{gilbert2013l2} and the third result of our paper require $k=n^{\Omega(1)}$. The constants $\gamma,\alpha$ should be thought as arbitrarily small constants, say $.001$. We also note that in the regime $k/\epsilon \leq n^{1-\alpha}$, the decoding time of our first algorithm becomes $(k/\epsilon) \log^{2+\gamma}n$. The exponents in the $\poly()$ factors in \cite{GLPS12} and \cite{gilbert2013l2} are at least $5$, though the authors did not attempt at an optimization of these quantities. The exponent in the $\poly()$ factors in \cite{larsen2016heavy} is at least $3$.
}\label{tbl:sparse}

\end{table}

\section{Our Techniques}


\subsection{Heavy hitters} All the schemes we give are for the strict turnstile model. Our first algorithm is based on a small but important tuning to the \textsc{Count-Min} sketch: we change the amount of independence in the hash functions from $2$-wise to $\Oh(1/\epsilon)$-wise. Since the estimate of any coordinate is an overestimate of it, we are able to show that the set of $\Oh(1/\epsilon)$ coordinates with the largest estimates is a superset of the set of the $\epsilon$-heavy hitters. Although changing the amount of independence might increase both the update and the query time by a multiplicative factor of $1/\epsilon$, we show that using fast multipoint evaluation of polynomials, we can suffer only a  $\log^2 (1/\epsilon)$ factor in the update time in the amortized case, and a $\log^2 (1/\epsilon)$ factor in query time in the worst case. The above observation for \textsc{Count-Min} sketch gives also an improvement on the well-known dyadic trick which appeared in the seminal paper of Cormode and Hadjieleftheriou \cite{cormode2008finding}.

For the deterministic case, our improved analysis of the \textsc{Count-Min} sketch implies a deterministic algorithm that finds heavy hitters of all vectors $x \in \mathbb{R}_{+}^n$; moreover, we show how expanders that expand only on sets of size $\Theta(1/\epsilon)$ (or equivalently lossless condensers) lead to schemes in the strict turnstile model
under Guarantee 3. Then we instantiate the Guruswami-Umans-Vadhan expander~\cite{guruswami2009unbalanced} properly to obtain an explicit algorithm. The idea of using expanders in the context of heavy hitters has been employed by Ganguly~\cite{ganguly2008data}, although his result was for the $\ell_{\infty}/ \ell_1$ problem with $\Omega(1/\epsilon^2)$ space. Known constructions of these combinatorial objects are based on list decoding, and do not achieve optimal parameters. Any improvement on explicit constructions of these objects would immediately translate to an improved explicit heavy hitters scheme for the strict turnstile model.

Our deterministic lower bounds are based on choosing ``bad input vectors'' for the sketching matrix $S$ based on several properties of $S$ itself. Since the algorithm is deterministic, it must succeed even for these vectors.  

Our randomized lower bounds come from designing a pair of distributions which must be distinguished by a heavy hitters algorithm with the appropriate guarantee. They are based on distinguishing a random Gaussian input from a random Gaussian input with a large coordinate planted in a uniformly random position. By Lipschitz concentration of the $\ell_1$-norm and $\ell_2$-norm in Gaussian space, we show that the norms in the two cases are concentrated, so we have a heavy hitter in one case but not the other. Typically, the planted large coordinate corresponds to a column in $S$ of small norm, which makes it indistinguishable from the noise on remaining coordinates. The proof is carried out by a delicate analysis of the total variation distance of the distribution of the image of the input under the sketch in the two cases.

\subsection{Non-Adaptive Sparse Recovery} 
Our result follows a similar framework to that of \cite{GLPS12}, though chooses more carefully the main primitives it uses and balances the parameters in a more effective way. Both schemes consist of $\Oh(\log k)$ so-called weak systems: a scheme that takes as input a vector $x$ and returns a vector $\hat{x}$ which contains a $2/3$ fraction of the heavy hitters of $x$ (the elements with magnitude larger than $\frac{1}{\sqrt{k}} \|x_{-k}\|_2$) along with accurate estimates of (most of) them. Then it proceeds by considering the vector $x-\hat{x}$, which contains at most $1/3$ of the heavy hitters of $x$. We then feed the vector $x-\hat{x}$ to the next weak-level system to obtain a new vector which contains at most $(2/3)(1/3)k = 2k/9$ of the heavy hitters. We proceed in a similar fashion, and after the $i$-th stage we will be left with at most $k/3^i$ heavy hitters.

Each weak system consists of an identification and an estimation part. The identification part finds a $2/3$ fraction of the heavy hitters while the identification part estimates their values. For the identification part, the algorithm in \cite{GLPS12} hashes $n$ coordinates to $\Theta(k)$ buckets using a $2$-wise independent hash function and then uses an error-correcting code in each bucket to find the heaviest element. Since, with constant probability, a heavy hitter will be isolated and its value will be larger than the `noise' level in the bucket it is hashed to, it is possible to find a $2/3$ fraction of the heavy hitters with constant probability and $\Oh(k \poly(\log n))$ decoding time. Moreover, in each bucket we use a $b$-tree, which is a folklore data structure in the data stream literature, the special case of which ($b=2)$ first appeared in \cite{cormode2005improved}. The estimation part is a different analysis of the folklore \textsc{Count-Sketch} data structure: we show that the estimation scheme can be implemented with an optimal dependence on $\epsilon$, in contrast to the the expander-based scheme in \cite{gilbert2013l2}, which gave a sub-optimal dependence on $\epsilon$.

In this paper, we first design an algorithm with running time $\Oh(k^2 \poly(\log n))$, as in \cite{gilbert2013l2}, and then improve it to $\Oh(k^{1+\alpha} \poly(\log n))$ time, but with a slightly larger failure probability. The key observation is that in the first round we find a constant fraction of the heavy hitters with $e^{-k}$ failure probability, in the second round we find a constant fraction of the remaining heavy hitters with $e^{-k/2}$ failure probability, and so on, with polynomially decreasing number of measurements. In later rounds we can decrease the number of measurements at a slower rate so that the failure probability can be reduced by using more measurements while the optimality of the number of measurements is retained. Our suffering from the quadratic dependence on $k$ in the runtime is due to the fact that our sensing matrix is very dense, with $\Oh(k)$ non-zeros per column. Hence updating measurements $y \gets y - \Phi\hat x$ will incur a running time proportional to $\|\hat x\|_0\cdot k$, where $\hat x$ is a $\Oh(k)$-sparse vector.\footnote{We note an omission in the runtime analysis in~\cite{gilbert2013l2}. Their measurement matrix contains $s=2^i/i^c$ (where $c$ is a constant) repetitions of an expander-based identification matrix (see Lemma 4.10 and Theorem 4.9 of~\cite{gilbert2013l2}). 
Each repetition has at least one non-zero entry per column and thus the measurement matrix for the $i$-th iteration has at least $s$ non-zero entries per column, which implies that when $i = \log k-1$, each column has at least $\Omega(k/\poly(\log k))$ nonzero entries. Updating measurements $y \gets y - \Phi\hat x$ will then take $\Omega(k^2/\poly(\log k))$ time, where $\hat x$ has $\Omega(k)$ nonzero coordinates. Therefore we would expect that the overall running time of the recovery algorithm will be $\tilde\Omega(k^2)$ instead of their claimed $\tilde\Oh(k^{1+\alpha})$.} 

But, how do we achieve an almost linear time algorithm while beating the constant failure probability of \cite{GLPS12}? The idea is to use again the same analysis, but without sharpening the failure probability in the first $(1/2) \log k$ steps. The first $(1/2)\log k$ rounds still fail with tiny failure probability, and once we reach round $(1/2) \log k$, we can afford to run the quadratic-time algorithm above, since our sparsity is now $\Oh(\sqrt{k})$. Hence we would expect the total algorithm to run in time $\Oh( \sqrt{k}^2 \poly(\log n)) = \Oh(k \poly(\log n))$. Putting everything together, we obtain substantial improvements over both \cite{GLPS12} and \cite{gilbert2013l2}. 

A caveat of our approach, which is the reason we obtain $k^{1+\alpha}$ dependence on the decoding time is the following. Since we do not want to store the whole matrix, our algorithms are implemented differently when $ k \leq n^{1-\alpha/2}$ and $k \geq n^{1-\alpha/2}$. In the former case, we use $\log n = \Theta(\log(n/k))$ measurements per bucket in the identification step, in order to avoid inverting an $O(k)$-wise independent hash function. In the latter case, to compute the pre-image of an $O(k)$-wise independent hash function we just evaluate the hash function, which corresponds to a degree $O(k)$ polynomial, in all places in time $O(n \log^2 k)$, and trivially find the pre-images. The asymptotic complexity of our algorithm in its full generality is dominated by the latter case, where we obtain $\tilde{\Oh}(k^{1+\alpha})$ decoding time. We note that in the regime $k \leq n^{1-\alpha}$, the running time becomes $\Oh(k \log^{2+\gamma} n)$.

\subsection{Adaptive Compressed Sensing} 
We start by implementing a $1/\poly(\log n)$ failure probability version of the $1$-sparse routine of \cite{indyk2011power}. We apply a preconditioning step before running \cite{indyk2011power} with a different setting of parameters; this preconditioning step gives us power for the next iteration, enabling us to achieve the desired failure probability in each round. 

The lemma above leads to a scheme for $\ell_2/\ell_2$ in the low-sparsity regime, when $k < \poly(\log n)$. The algorithm operates by hashing into $\poly(\log n)$ buckets, determining the heavy buckets using a standard variant of \textsc{Count-Sketch}, and then running the $1$-sparse recovery in each of these buckets. The improved algorithm for $1$-sparse recovery is crucial here since it allows for a union bound over all buckets found. 

For the case of general $k$-sparsity, we show that the main iterative loop of \cite{indyk2011power} can be modified so that it gives exponentially smaller failure probability in $k$. The idea is that, as more and more heavy hitters are found, it is affordable to use more measurements to reduce the failure probability. Interestingly and importantly for us, the failure probability per round is minimized in the first round, and in fact is increasing exponentially, although this was not exploited in \cite{indyk2011power}. Therefore, in the beginning we have exponentially small failure probability, but in later rounds we can use more measurements to boost the failure probability by making more repetitions. This part needs care in order not to blow up the number of measurements while achieving the best possible failure probability. We use a martingale argument to handle the dependency issue that arises from hashing coordinates into buckets, and thus avoid additional repetitions that would otherwise increase the number of of measurements.

The two algorithms above show how we can beat the failure probability of \cite{indyk2011power} for all values of $k$: we have $1/\poly(\log n)$ for small $k$ and  $e^{k^{-0.999}}$ for large $k$, thus achieving asymptotic improvements in every case.

We note that although in the heavy hitters schemes we take into account the space to store the hash functions, in sparse recovery we adopt the standard practice of not counting the space needed to store the measurement matrix, and therefore we use full randomness.



\section{Formal Statement of Results}

In this section we state all of our results and in subsequent sections we shall only give an outline of our improved analysis of \textsc{Count-Min} and our lower bound for \textsc{Count-Sketch}. The proofs of all other theorems can be found in the appendix. 
The notations $\Oh_{a,b,\dots}, \Omega_{a,b,\dots}$ indicate that the constant in $\Oh$- and $\Omega$-notations depend on $a,b,\dots$.

\subsection{Heavy Hitters}

\subsubsection{Upper Bounds}
\begin{theorem}[$\ell_1$ Heavy Hitters Under Guarantee 3]\label{thm:promise-body}
There exists a data structure $\mathrm{DS}$ which finds the $\ell_1$ heavy hitters of any $x \in \mathbb{R}^n$ in the strict turnstile model under Guarantee 3. In other words, we can sketch $x$, such that we can find a list $L$ of $\Oh(k)$ coordinates that contains all $\eps$-heavy hitters of $x$. The space usage is $\Oh( \frac{1}{\epsilon} \log (\epsilon n))$, the update time is amortized $\Oh(\log^2(\frac{1}{\epsilon}) \log (\epsilon n) )$ and the query time is $\Oh( n \log^2(\frac{1}{\epsilon})  \log (\epsilon n ) )$.

\end{theorem}

The following theorem follows by an improved analysis of the dyadic trick \cite{cormode2008finding}.

\begin{theorem}There exists a data structure with space $\Oh(\frac{1}{\epsilon} \log (\epsilon n) + \log(\epsilon n) \cdot \log  ( \frac{\log (\epsilon n)}{\delta}))$ that finds the $\ell_1$ heavy hitters of $x \in \mathbb{R}^n$ in the strict turnstile model under Guarantee 3 with probability at least $1 - \delta$. The update time is $\Oh( \log^2(\frac{1}{\epsilon}) \log (\epsilon n) +  \epsilon \log(\epsilon n) \log^2(\frac{1}{\epsilon}) \log(\frac{\log( \epsilon n)}{\delta}))$ amortized and the query time is $\Oh( \frac{1}{\epsilon}( \log^2(\frac{1}{\epsilon}) \log (\epsilon n) +  \log(\epsilon n) \log^2(\frac{1}{\epsilon}) \log (\frac{\log(\epsilon n) }{\delta})) ))$.
\end{theorem}

\begin{theorem}[Explicit $\ell_1$ Heavy Hitters in the Strict Turnstile Model]
There exists a fully explicit algorithm that finds the $\epsilon$-heavy hitters of any vector $x \in \mathbb{R}^n$ using space $\Oh(k^{1+ \alpha} (\log (\frac{1}{\epsilon}) \log n)^{2+ 2/ \alpha})$. The update time is $\Oh( \mathrm{poly}(\log n))$ and the query time is $\Oh( n \cdot \mathrm{poly}( \log n))$.
\end{theorem}



\subsubsection{Lower Bounds}

\begin{theorem}[Strict turnstile deterministic lower bound for Guarantees 1,2]\label{thm:det1-body}
  Assume that $n = \Omega(\epsilon^{-2})$. Any sketching matrix $S$ must have $\Omega(\epsilon^{-2})$ rows if, in the strict turnstile model,
  it is always possible to recover from $Sx$ a set which contains all the $\epsilon$-heavy hitters of $x$ and contains no items which are not $(\epsilon/2)$-heavy hitters. 
  \end{theorem}

\begin{theorem}[Turnstile deterministic lower bound for Guarantee 3]\label{thm:det2-body} 
 Assume that $n = \Omega(\epsilon^{-2})$. Any sketching matrix $S$ must have $\Omega(\epsilon^{-2})$ rows if, in the turnstile model, some algorithm never fails in returning a superset of size $O(1/\epsilon)$ containing the $\epsilon$-heavy hitters. Note that it need not return approximations to the values of the items in the set which it returns.
\end{theorem}

\begin{theorem}[Randomized Turnstile $\ell_1$-Heavy Hitters Lower Bound for Guarantees $1,2$]\label{thm:randomizedlb}
  Assume that
  $1/\epsilon \geq C \sqrt{\log(1/\delta)}$ and suppose $n \geq  \left \lceil 64 \epsilon^{-1} \sqrt{\log(1/\delta)} \right \rceil$. Then
    for any sketching matrix $S$, it must have $\Omega(\epsilon^{-1} \sqrt{\log(1/\delta)})$ rows if, in the turnstile model, it succeeds
    with probability at least $1-\delta$ in returning a set containing all the $\epsilon$ $\ell_1$-heavy hitters and containing no
    items which are not $(\epsilon/2)$ $\ell_1$-heavy hitters. 
\end{theorem}

\begin{theorem}[Randomized $\ell_2$-Heavy Hitters Lower Bound with Guarantee 2]\label{thm:ell2_hhlb}
    Suppose that $\delta < \delta_0$ and $\epsilon < 1/\epsilon_0$ for sufficiently small absolute constants $\delta_0,\epsilon_0\in (0,1)$ and $n \geq  \left \lceil 64 \epsilon^{-1} \log(6/\delta) \right \rceil$. Then
    for any sketching matrix $S$, it must have $\Omega(\epsilon^{-1} \log(1/\delta))$ rows if it succeeds
    with probability at least $1-\delta$ in returning a set containing all the $\epsilon$-heavy hitters and containing no
    items which are not $(\epsilon/2)$-heavy hitters. 
\end{theorem}

\subsection{Non-Adaptive Sparse Recovery}
\begin{theorem}\label{thm:nonadaptive_quadratic_runtime}
Let $1\leq k\leq n$ be integers and $\gamma > 0$ be a constant. There exists an $\ell_2/\ell_2$ sparse recovery system $\mathcal{A} = (\mathcal{D},\mathcal{R})$ with parameters $(n,k,\epsilon,\Oh_\gamma(k/\epsilon \log(n/ \epsilon k)),\exp(-k/\log^3 k))$. Moreover, $\mathcal{R}$ runs in time $\Oh_\gamma(k^2 \log^{2+\gamma}n)$.

In other words, there exists an  $\ell_2/\ell_2$ sparse recovery system that uses $\Oh(k \epsilon \log(n/\epsilon k))$ measurements, runs in time $\Oh_{\gamma}(k^2 \log^{1+\gamma}n)$, and fails with probability $\exp(-k/\log^3 k)$. 
\end{theorem} 

\begin{theorem}\label{thm:nonadaptive_linear_runtime}
Let $1\leq k\leq n$ be integers and $\gamma > 0$ be a constant. There exists an $\ell_2/\ell_2$ sparse recovery system $\mathcal{A} = (\mathcal{D},\mathcal{R})$ with parameters $(n,k,\epsilon,\Oh_\gamma(k/\epsilon \log(n/k)),\exp(-\sqrt{k}/\log^3 k))$. Moreover, $\mathcal{R}$ runs in time $\Oh_\gamma(k/\epsilon \log^{2+\gamma}n)$. In other words, there exists an  $\ell_2/\ell_2$ sparse recovery system that uses $\Oh(k \epsilon \log(n/\epsilon k))$ measurements, runs in time $\Oh_{\gamma}(k \log^{2+\gamma}n)$, and fails with probability $\exp(-\sqrt{k}/\log^3 k)$. 
\end{theorem}

\begin{theorem}
Suppose that $k = n^{\Omega(1)}$. There exists an $\ell_2/\ell_2$ sparse recovery system $\mathcal{A}= (\mathcal{D}, \mathcal{R})$ with parameters $\left(n,k,\epsilon, \Oh(\frac{k}{\epsilon} \log\frac{n}{\epsilon k}), (\frac{n}{k})^{-\frac{k}{\log k}}\right)$. Moreover, $\mathcal{R}$ runs in $\Oh(k^2/\epsilon \poly(\log n))$ time. In other words, there exists an  $\ell_2/\ell_2$ sparse recovery system that uses $\Oh(k \epsilon \log(n/\epsilon k))$ measurements, runs in time $\Oh(k^2/\epsilon \poly(\log n))$, and fails with probability $(n/k)^{-k/\log k}$. 

\end{theorem}

\subsection{Adaptive Sparse Recovery}

\begin{theorem}[Entire regime of parameters] \label{thm:adaptive_whole_regime}
Let $ x \in \mathbb{R}^n$ and $\gamma>0$ be a constant. There exists an algorithm that performs $\Oh( (k/\epsilon) \log \log(\epsilon n/k))$ adaptive linear measurements on $x$ in $\Oh( \log^*k \cdot \log \log ( \epsilon n/k))$ rounds, and finds a vector $\hat{x} \in \mathbb{R}^n$ such that $\| x - \hat{x}\|_2^2 \leq (1+\epsilon) \|x_{-k}\|_2^2$. The algorithm fails with probability at most $\exp(-k^{1-\gamma})$.
\end{theorem}

\begin{theorem}[low sparsity regime]
Let $ x \in \mathbb{R}^n$ and parameters $k,\epsilon$ be such that $k/\epsilon \leq c \log n$, for some absolute constant $c$. There exists an algorithm that performs $\Oh( (k/\epsilon) \log \log n)$ adaptive linear measurements on $x$ in $\Oh( \log \log n)$ rounds, and finds a vector $\hat{x} \in \mathbb{R}^n$ such that $ \| x - \hat{x}\|_2 \leq (1+\epsilon) \|x_{-k/\epsilon} \|_2$. The algorithm fails with probability at most $1/\poly(\log n)$.
\end{theorem}

\subsection{Spiked Covariance Model}

In the spiked covariance model, the signal $x$ is subject to the following distribution: we choose $k$ coordinates uniformly at random, say, $i_1,\dots,i_k$. First, we construct a vector $y\in \R^n$, in which each $y_{k_i}$ is a uniform Bernoulli variable on $\{-\sqrt{\epsilon/k}, +\sqrt{\epsilon/k}\}$ and these $k$ coordinate values are independent of each other. Then let $z\sim N(0,\frac{1}{n}I_n)$ and set $x = y+z$. We now present a non-adaptive algorithm (although the running time is slow) that uses $\Oh((k/\epsilon)\log(\epsilon n/k) + (1/\epsilon)\log(1/\delta))$ measurements and present a matching lower bound.

\begin{theorem}[Upper Bound]
Assume that $(k/\epsilon)\log(1/\delta) \leq \beta n$, where $\beta\in(0,1)$ is a constant. There exists an $\ell_2/\ell_2$ algorithm for the spiked-covariance model that uses $\Oh\left(\frac{k}{\epsilon}\log\frac{\epsilon n}{k}+ \frac{1}{\epsilon}\log\frac{1}{\delta}\right)$ measurements and succeeds with probability $\geq 1-\delta$. Here the randomness is over both the signal and the algorithm.
\end{theorem}

\begin{theorem}[Lower Bound]
    Suppose that $\delta < \delta_0$ for a sufficiently small absolute constant $\delta_0\in (0,1)$ and $n \geq  \left \lceil 64 \epsilon^{-1} \log(6/\delta) \right \rceil$. Then any $\ell_2/\ell_2$-algorithm that solves with probability $\geq 1-\delta$ the $\ell_2/\ell_2$ problem in the spiked-covariance model must use $\Omega(\epsilon^{-1} \log(1/\delta))$ measurements. 
\end{theorem}

Combining with the lower bound from \cite{PW11} (Section 4) we get a lower bound for the spiked covariance model of $\Omega((k/\epsilon) \log(n/k) + \log(1/\delta)/\epsilon)$. We note that although the lower bound is not stated for the spiked covariance model, inspection of the proof indicates that the hard instance is designed in that model.

\section{Improved Analysis of \textsc{Count-Min} sketch}
We sketch the proof of Theorem~\ref{thm:promise-body} in this section. We shall first show a randomized algorithm with failure probability $\delta$ and the theorem then follows from taking a union bound over all possible subsets of size $\Oh(1/\epsilon)$ that contain the heavy hitters. The randomized algorithm is presented in Algorithm~\ref{alg:count-min-body}, where we set $C_{R} =5, C_{\delta} = 10(\ln 4 -1), C_B = 20, C_0= 30$.

\begin{algorithm}
\caption{\textsc{Count-Min}: Scheme for Heavy Hitters in the Strict Turnstile Model}\label{alg:count-min-body}
\begin{algorithmic}
\State $R \gets C_{R} \log (\epsilon n) + \lceil \frac{ \epsilon \log(1/\delta)}{C_{\delta}} \rceil$
\State $B \gets C_B k /\epsilon$
\State Pick hash functions $h_1, \ldots , h_{R}: [n] \rightarrow [B]$, each being $(C_0/\epsilon)$-wise independent
\State $C_{r,j}\gets 0$ for all $r \in [R]$ and $j \in [B]$.

\medskip

\Function{Update}{$i,\Delta$}
\State $C_{r, h_r(i)} \leftarrow C_{r,h_r(i)} + \Delta$ for all $r \in [R]$
\EndFunction

\medskip

\Function{Query}{}
\State $\hat{x}_i \gets  \min_r C_{r,h_r(i)}$ for all $i\in [n]$
\State \Return the $(C_0+1) \frac{1}{\epsilon}$ coordinates in $P$ with the largest $\hat{x}$ values.
\EndFunction
\end{algorithmic}
\end{algorithm}

\begin{proof}[Proof of Theorem~\ref{thm:promise-body} (Sketch)]
We prove first that Algorithm $1$ satisfies Guarantee $3$ with probability $\geq 1-\delta$. Let $S$ be the set of $\epsilon$-heavy hitters of $x$.  Let $T$ be any set of at most $C_0/\epsilon$ coordinates that is a subset of $[n]\setminus S$. Fix a hash function $h_a$. Let $B_a[T]$ denote the set of the indices of the buckets containing elements from $T$, i.e.\@ $B_a[T] = \{ b \in [B]: \exists j \in T, h_a(j) = b\}$. The probability that $|B_a[T]| <10/\epsilon $ is at most 
\[
\binom{B}{10/\epsilon}\left(\frac{10/\epsilon}{B}\right)^{\frac{C_0}{\epsilon}} \leq \left(\frac{e C_Bk}{10 }\right)^{\frac{10}{\epsilon}} \left(\frac{10}{C_Bk}\right)^{\frac{C_0}{\epsilon}}  \leq \left(\frac{e}{4}\right)^{\frac{10}{\epsilon}}.
\] 
Considering all $R$ hash functions, it follows that with probability $\geq 1 - (e/4)^{10R/\epsilon}$ there exists an index $ a_T^* \in [R]$ such that $|B_{a_T^*}[T]| \geq 10/\epsilon$. A union bound over all possible $T$ yields that with probability at least $1- \delta$, there exists such an index $a^*_T$ for each $T$ of size $C_0 /\epsilon$.

Now, let $Out$ be the set of coordinate the algorithm outputs and let $T' = Out\setminus S$. Clearly, $|T'| \geq C_0/\epsilon$. We can discard coordinates of $T'$ so that $|T'| = C_0/\epsilon$. We shall prove that there exists an element $j \in T'$ such that its estimate is strictly less than $\epsilon \|x_{-1/\epsilon }\|_1$ and hence smaller than the estimate of any element in $S$. This will imply that every element in $S$ is inside $Out$. From the previous paragraph, there exists $a_{T'}^*$ such that $|B_{a_{T'}[T']}| > 10/\epsilon$. Since we have at most $1/\epsilon$ heavy hitters and at most $1/\epsilon$ indices $b \in [B]$ such that the counter $ C_{a,b} \geq \epsilon  \|x_{-1/\epsilon} \|_1$, at least $10/\epsilon$ indexes (buckets) of $B$ such that the bucket has mass less than $\epsilon \|x_{1/\epsilon}\|_1$ while at least one element of $T'$ is hashed to that bucket. Therefore, the estimate for this element is less than $\epsilon \|x_{- 1/\epsilon}\|_1$, which shows correctness.

Algorithm $1$ clearly has an update time of $\Oh( \frac{1}{\epsilon} \log(\epsilon n) + \log\frac{1}{\delta})$ and a query time $\Oh(\frac{1}{\epsilon} n \log(\epsilon n) + n\log\frac{1}{\delta})$. In the appendix, we have the full proof showing how to improve the update time to amortized
$\Oh(  \log (\epsilon n) \log^2\frac{1}{\epsilon} +   \epsilon \log \frac{1}{\delta} \log^2\frac{1}{\epsilon} )$  
 and the query time to $\Oh(n (\log n +\epsilon\log \frac{1}{\delta} )\log^2\frac{1}{\epsilon})$, via fast multi-point evaluation of polynomials.


Finally, let $\delta = 1/\binom{n}{1/\epsilon}$ so that we can take a union bound over all possible subsets of heavy hitters.
\end{proof}

\section{$\ell_{\infty}/\ell_2$ lower bound}

This section is devoted to the proof of Theorem \ref{thm:ell2_hhlb}. The proof is based on designing a pair of hard distributions which
cannot be distinguished by a small sketch. We show this 
by using rotational properties of the Gaussian distribution to reduce
our problem to a univariate Gaussian mean estimation problem, which we
show is hard to solve with low failure probability.

\subsection{Toolkit}
We list some facts in measure of concentration phenomenon in this subsection and defer all the proofs to Appendix~\ref{sec:lb_prelim}.

\begin{fact}[Total Variation Distance Between Gaussians]\label{lem:gtvdMain}
  $
  D_{TV}(N(0, I_r), N(\tau, I_r)) = \Pr_{g\sim N(0,1)}\left\{|g| \leq \frac{\|\tau\|_2}{2}\right\}.
  $
\end{fact}


\begin{fact}[Concentration of $\ell_2$-Norm]\label{fact:small2NormMain}
  Suppose $x \sim N(0, I_n)$ and $n \geq 18\ln(6/\delta)$. Then\\
  $\Pr \left\{\frac{\sqrt{n}}{2} \leq \|x\|_2 \leq \frac{3\sqrt{n}}{2} \right\} \geq 1-\frac{\delta}{3}$.
\end{fact}

  \begin{fact}[Univariate Tail Bound]\label{fact:coordinateMain}
    Let $g\sim N(0,1)$. There exists $\delta_0 > 0$ such that it holds for all $\delta < \delta_0$ that 
    $\Pr\{|g| \leq 4 \sqrt{\log(1/\delta)}\} \geq 1-\delta/3$. 
    \end{fact}
  
In our proofs we are interested in lower bounding the number $r$ of rows of a sketching matrix $S$.

\subsection{Proof of the lower bound}

\begin{proof}
Let the universe size be $n = \left \lceil 64\epsilon^{-1} \log(6/\delta) \right \rceil,$ which is large
enough in order for us to apply Fact \ref{fact:small2NormMain} (note if the actual universe size is larger, we can set all but
the first $\left \lceil 64\epsilon^{-1} \log(6/\delta) \right \rceil$ coordinates of our input to $0$). 

Let $r$ be the number of rows of the sketch matrix $S$, where $r \le n$. If $r > n$, then we immediately
obtain an $\Omega(\epsilon^{-1} \log(1/\delta))$ lower bound. We can assume that
$S$ has orthonormal rows, since a change of basis to the row space of $S$
can always be performed in a post-processing step.

\medskip\noindent\textbf{Hard Distribution.} Let $I$ be a uniformly random index in $[n]$. 

\noindent{\bf Case 1:}
Let $\eta$ be the distribution $N(0, I_n)$, and suppose $x \sim \eta$. By Fact \ref{fact:small2NormMain},
$\|x\|_2 \geq \sqrt{n/2}$ with probability $1-\delta/3$.
By Fact \ref{fact:coordinateMain}, $|x_I| \leq 4 \sqrt{\log(1/\delta)}$ with probability $1-\delta/3$. Let $\mathcal{E}$
be the joint occurrence of these events, so that $\Pr(\mathcal{E}) \geq 1-2\delta/3$. 

By our choice of $n$, it follows that if $\mathcal{E}$ occurs, then $x_I^2 \leq 16 \log(1/\delta) \leq \frac{\epsilon}{2} \|x\|_2^2$, 
and therefore $I$ cannot be output by
an $\ell_2$-heavy hitters algorithm. 

\noindent{\bf Case 2:} Let $y\sim N(0,I_n)$ and $x = \sqrt{\epsilon n} e_I + y$, where $e_I$ denotes the standard basis vector in the $I$-th direction. 
By Fact \ref{fact:small2NormMain}, $\|y\|_2 \leq \frac{3\sqrt{n}}{2}$ with probability $1-\delta/3$.
By Fact \ref{fact:coordinateMain}, $|y_I| \leq 4 \sqrt{\log(1/\delta)} < \sqrt{\epsilon n}/2$ with probability $1-\delta/3$. Let $\mathcal{F}$
be the joint occurrence of these events, so that $\Pr(\mathcal{F}) \geq 1-2\delta/3$. 

If event $\mathcal{F}$ occurs, then $|x_I| \geq 3\sqrt{\epsilon n}- 4\sqrt{\log(1/\delta)} \geq \frac{5 \sqrt{\epsilon n}}{2}$, and
so $x_I^2 \geq \frac{25 \epsilon n}{4}$. We also have 
$\|x\|_2 \leq 3\sqrt{\epsilon n} + \frac{3\sqrt{n}}{2} \leq 2\sqrt{n}$, provided $\epsilon \leq 1/36$, and so
$\|x\|_2^2 \leq 4 n$. Consequently, $x_I^2 \geq \epsilon \|x\|_2^2$. 
Consequently, if $\mathcal{F}$ occurs,
for an $\ell_2$-heavy hitters algorithm to be correct, it must output $I$. 

\medskip\noindent\textbf{Conditioning.}
Let $\eta'$ be the distribution of $\eta$ conditioned on $\mathcal{E}$,
and let $\gamma'$ be the distribution of $\gamma$ conditioned on $\mathcal{F}$.
For a distribution $\mu$ on inputs $y$, we let $\bar{\mu}$ be the distribution of $Sy$.

Note that any $\ell_2$-heavy hitters algorithm which succeeds with probability at least $1-\delta$ can decide,
with probability at least $1-\delta$, whether $x \sim \eta'$ or $x \sim \gamma'$.
Hence, $D_{TV}(\bar{\eta'}, \bar{\gamma'}) \geq 1-\delta$. Observe for any measurable set $A\subseteq \R^m$ it holds that 
\[
\left|\frac{\bar\eta(A)-\bar\mu(A)}{1-\frac23\delta} - (\bar\eta'(A)-\bar\mu'(A))\right| \leq \frac{2}{3}\delta,
\]
and so it then follows that
$
D_{TV}(\bar{\eta}, \bar{\gamma}) \geq \left(D_{TV}(\bar{\eta'}, \bar{\gamma'}) - \frac{2\delta}{3}\right)\left(1-\frac{2}{3}\delta\right)  \geq 1-\frac{7\delta}{3}$.
Therefore, to obtain our lower bound, it suffices to show if the number $r$ of rows of $S$ is too small, then it cannot hold that
$D_{TV}(\bar{\eta}, \bar{\gamma}) \geq 1-7\delta/3$.

\medskip\noindent\textbf{Bounding the Total Variation Distance.}
Since $S$ has orthonormal rows, by rotational invariance of the Gaussian distribution, the distribution of $\bar{\eta}$
is identical to $N(0, I_r)$ and the distribution
of $\bar{\gamma}$ identical to $(3 \sqrt{\epsilon n}) S_I + N(0, I_r)$, where $S_I$ is the $I$-th column of $S$. 

Since $S$ has orthonormal rows,
by a Markov bound, for $9/10$ fraction of values of $I$, it holds that $\|S_I\|_2^2 \le \frac{10r}{n}$.
Call this set of columns $T$. 

Let $\mathcal{G}$ be the event that $I \in T$, then $\Pr(\mathcal{G}) \geq 9/10$. It follows that
\begin{align*}
D_{TV}(\bar\eta,\bar\gamma) \leq \Pr(G)D_{TV}(\bar\eta,\bar\gamma|\mathcal{G}) + \Pr(\neg G)D_{TV}(\bar\eta,\bar\gamma|\neg\mathcal{G}) &\leq \Pr(G)D_{TV}(\bar\eta,\bar\gamma|\mathcal{G}) + 1-\Pr(G)\\
&= 1 - \Pr(G)(1-D_{TV}(\bar\eta,\bar\gamma|\mathcal{G}))\\
&\leq 1-\frac{9}{10}(1-D_{TV}(\bar\eta,\bar\gamma|\mathcal{G})).
\end{align*}
Hence, in order to deduce a contradiction that $D_{TV}(\bar\eta,\bar\gamma)< 1-7\delta/3$, it suffices to show that $D_{TV}(\bar\eta,\bar\gamma|\mathcal{G})< 1-70\delta/27$.

The total variation distance between $N(0, I_r)$ and $(3 \sqrt{\epsilon n}) S_i + N(0, I_r)$ for a fixed $i \in T$ is,
by rotational invariance and by
rotating $S_i$ to be in the same direction as the first standard basis vector $e_1$,
the same as the total variation distance between $N(0, I_r)$ and 
$(3 \sqrt{\epsilon n}) \|S_i\|_2 e_1 + N(0, I_r),$
which is equal to the total variation distance between $N(0, 1)$ and $N(3 \sqrt{\epsilon n} \|S_i\|_2, 1)$.

Using that $i \in T$ and so $\|S_i\|_2 \leq \sqrt{10r/n}$, we apply Fact \ref{lem:gtvdMain} to obtain that
the variation distance is at most $\Pr[|N(0,1)| \leq (3/2) \sqrt{\epsilon n} \cdot \sqrt{10r/n}]$. It follows that
\[
  D_{TV}(\bar{\eta}, \bar{\gamma} \mid \mathcal{G}) \leq  \sum_{i \in T} \frac{1}{|T|} D_{TV}(\bar{\eta}, \bar{\gamma} \mid I = i)
  \leq \Pr_{g\sim N(0,1)}\left\{|g| \leq (3/2) \sqrt{\epsilon n} \cdot \sqrt{10r/n}\right\},
\]
and thus it suffices to show, when $r$ is small, that
\[
\Pr_{g\sim N(0,1)}\left\{|g| \geq \frac{3}{2}\sqrt{10 \epsilon r}\right\} > \frac{70\delta}{27}.
\]
Observe that the left-hand is a decreasing function in $r$, and so it suffices to show the inequality above for $r = \alpha\epsilon^{-1}\log(1/\delta)$ for some $\alpha > 0$.

Invoking the well-known bound that
(see, e.g., \cite{gordon1941})
\[
\Pr_{g\sim N(0,1)}\left\{g \geq t\right\} \geq \frac{1}{\sqrt{2\pi}}\cdot\frac{1}{2t}e^{-\frac{t^2}{2}}, \quad t\geq \sqrt{2},
\]
we have that
\[
\Pr_{g\sim N(0,1)}\left\{|g| \geq \frac{3}{2}\sqrt{10 \epsilon r}\right\}\geq \frac{3}{4\sqrt{5\pi}}\delta^{\frac{45}{4}\alpha}\frac{1}{\sqrt{\alpha\log(1/\delta)}} > \frac{70}{27}\delta
\]
when $\alpha$ is small enough. Therefore it must hold that $r \geq \alpha\epsilon^{-1} \log(1/\delta)$ and the proof is complete.
\end{proof}

\bibliographystyle{plainurl}
\bibliography{biblio}

\appendix

\section{Additional Preliminaries}


\noindent\textbf{Count-Sketch.} Our estimation will be based on the \textsc{Count-Sketch} data structure. The \textsc{Count-Sketch} of a vector $x$ is defined as follows. For $r=1,\dots, R$, we hash all $n$ coordinates of $x$ into $B$ buckets with function $h_r:[n]\to [B]$. In each bucket we aggregate the elements with random signs, that is, the value of bucket $b\in[B]$ in iteration $r$ is
\[
V_{b,r} = \sum_{i: h_r(i)=b} \sigma_{i,r} x_i
\]
where the $\sigma_{i,r}$ are independent random signs ($\pm 1$). For $i\in I$ the estimate to $x_i$ returned by \textsc{Count-Sketch} is
\[
\hat{x}_i = \median_{1\leq r\leq R} \sigma_{i,r}V_{h_r(i),r}.
\]

\noindent\textbf{Weak System.} We follow the approach in several previous works~\cite{PS12,GLPSl1} by first constructing a weak system and then building an overall algorithm upon it.  We define $H_{k,\epsilon}(x) = \{ i \in [n]: |x_i|^2 \geq \epsilon/ k \|x_{-k}\|_2^2 \}$.

\begin{definition}
A probabilistic matrix $M$ with $n$ columns is called a $(k, \zeta)$-weak identification matrix with the $(\ell,\delta,\epsilon)$ guarantee if there is an algorithm that given $Mx$ and a subset $S \subseteq [n]$, with probability at least $1-\delta$ outputs a subset $I \subseteq S$ such that $|I| \leq \ell$ and at most $\zeta k $ elements of $H_{k,\epsilon}(x)$ are not present in $I$. The time to compute $I$ will be called the identification time.
\end{definition}

\begin{definition}
We call an $m \times n$ matrix $M$ a $(k,\zeta, \eta)$ weak $\ell_2/\ell_2$ system if the following holds for any vector $x= y+z$ such that $|\supp(y)| \leq k$: given $Mx$, one can compute $\hat{x}$ such that there exist $\hat{y}$, $\hat{z}$ which satisfy the following properties: (i) $x - \hat{x} = \hat{y} + \hat{z}$, (ii) $| \supp(\hat{x})| = \Oh(k)$, (iii) $|\supp(\hat{y})| \leq \zeta k$ and (iv) $\|\hat z\|_2^2 \leq (1+\eta)\|z\|_2^2$. \end{definition}

\subsection{Toolkit}\label{sec:toolkit}

\subsubsection{Black-Box Routines}

In this subsection we mention two theorems crucial for our sparse recovery scheme. The first one is from the recent paper by Larsen et al.\ on finding heavy hitters in data streams~\cite{larsen2016heavy}.

\begin{theorem} \label{thm:expandersketch}
Let $x \in \mathbb{R}^n$, $\delta>0$ and $K\geq 1$. There exists an (oblivious) randomized construction of a matrix $A$ such that given $y=Ax$ we can find, with probability $1-\delta$, a list $L$ of size $\Oh(K)$ that contains all $\frac{1}{K}$-heavy hitters of $x$ with probability $1-\delta$. The number of rows of $A$ is $\Oh(K \log (n/\delta))$, the time to find the list $L$ is $\Oh(K \log^3 n)$, and the column sparsity is $\Oh(\log (n/\delta))$. 
\end{theorem}

We also need the folklore $b$-tree data structure, which is a generalization of the dyadic trick. The next analysis of a $b$-tree-based heavy hitter algorithm appears in $\cite{larsen2016heavy}$.
 
\begin{theorem}\label{thm:btree}
For a vector $x$ which lives in a universe of size $n$, there exists a data structure produces a correct output $L$ for $(1/K)$-heavy hitters using space $\Oh(K \log n)$, having update time $\Oh( \log n)$ and query time $\Oh( (K^{1+\gamma/2}/\delta^\gamma) \log^{1+\gamma} n)$, and failing with probability $1-\delta$, when $\delta > 1/\mathrm{poly}(n)$. The constants in the big-Oh notations depend on $\gamma$.
\end{theorem}

We will also need the following theorem from \cite{gilbert2013l2}. Although the theorem in the paper is not stated this way, it is easy to see that by setting the quality of the approximation to be $1$ and the sparsity to be $k/\epsilon$ we immediately obtain the desired result.

\begin{theorem}[{\cite{gilbert2013l2}}] \label{thm:gilbert_identification}
Suppose that $k = n^{\Omega(1)}$. There exists a randomized construction of a matrix $M$ of $(k/\epsilon)\log(\epsilon n /k)$ rows, such that given $Mx$ one can find a set $S$ of size $\Oh(k/\epsilon)$ such that $| H_{k/\epsilon,1} \setminus S | \leq \zeta k$ with probability $1 - \binom{n}{k}^{-c}$, where $\zeta,c>0$ are absolute constants. The time to find $S$ is $\Oh(k^{1+\alpha} \poly( \log n))$, where $\alpha$ is any arbitrarily small positive constant. The constants in the $\Oh$-notation for the size of $S$ and the runtime depend on $\alpha$.
\end{theorem}

\subsubsection{Negative Dependence}

We review some basic results of negatively associated variables below. The definitions and propositions can be found in \cite{negativedependence}.
\begin{definition}[Negative association]
Random variables $\{X_i\}_{i\in [n]}$ are negatively associated (NA) if for every two disjoint sets $I,J \subseteq [n]$ and for all non-decreasing functions $f$ and $g$,
\[ \mathbb{E} \{f(X_i, i\in I)g(X_j, j \in J) \} \leq\mathbb{E} f(X_i, i\in I) \mathbb{E} g(X_j, j \in J). \]
\end{definition}
The following two propositions are useful in establishing negative association without verifying the definition above.
\begin{proposition}[Closure under products]\label{prop:closure} Let $X_1,\dots,X_n$ and $Y_1,\dots,Y_m$ be two independent families of random variables that are separately negatively associated, then the family $X_1,\dots,X_n,Y_1,\dots,Y_m$ is also negatively associated.
\end{proposition}
\begin{proposition}[Disjoint monotone aggregation]\label{prop:disjoint_aggr} Let $\{X_i\}_{i\in [n]}$ be negatively associated random variables and $\mathcal{A}$ is a family of disjoint subsets of $[n]$, then the random variables $\{f_A(X_i,i\in A)\}_{A\in \mathcal{A}}$ are also negatively associated, where $f_A$'s are arbitrary non-decreasing (or non-increasing) functions.
\end{proposition}

The following is a classical fact concerning the balls-into-bins model.
\begin{fact}\label{thm:ballsandbins}
Suppose that we hash $m$ balls independently to $n$ bins (different balls are not necessarily hashed with the same distribution) and let $B_{i,j}$ be the indicator random variable of the event `ball $j$ balls in bin $i$'. Then $B_{i,j}$ are negatively associated. As a consequence, the number of balls in the $i$-th bin, $B_i = \sum_j B_{i,j}$, are also negatively associated.
\end{fact}

\section{Heavy Hitters Problem}\label{sec:HH}

\subsection{Upper Bounds}

In this section we give schemes that enable the detection of heavy hitters in data streams deterministically or with low failure probability. Our first result is Theorem 1, which is slightly more general than finding heavy hitters but will be useful later when implementing the dyadic trick. This theorem refers to a variant of the problem which we call the Promise Heavy Hitters problem, where we are given a subset $P$ of $[n]$ of size $m$ with the guarantee that every heavy hitter is contained in the subset $P$.

\begin{theorem}\label{thm:promise}
Let $x \in \mathbb{R}^n$ be a vector with non-negative coordinates. Assume a promise that the heavy hitters of $x$ lie in a set $P$ of size $m$. There exists a data structure such that with probability $1 - \delta$, upon query, it finds the $\ell_1$ heavy hitters (i.e., $p = 1$) of $x$ under Guarantee 3. The space usage is $\Oh(  \frac{1}{\epsilon} \log(\epsilon m) + \log ( \frac{1}{\delta}))$, the amortized update time is $\Oh(  \log (\epsilon m) \log^2(\frac{1}{\epsilon}) +   \epsilon \log^2(\frac{1}{\epsilon}) \log( \frac{1}{\delta} ))$ and the worst-case query time is $\Oh( m  \log^2(\frac{1}{\epsilon}) (\log m +\epsilon\log( \frac{1}{\delta}) ))$.
\end{theorem}

As we shall see later, the above scheme implies the following corollary.

\begin{corollary}\label{thm:corollary}
There exists a data structure $\mathrm{DS}$ which finds the $\ell_1$ heavy hitters of any $x \in \mathbb{R}^n$ in the strict turnstile model under Guarantee 3. The space usage is $\Oh( \frac{1}{\epsilon} \log (\epsilon n))$, the update time is amortized $\Oh(\log^2(\frac{1}{\epsilon}) \log (\epsilon n) )$ and the query time is $\Oh( n \log^2(\frac{1}{\epsilon})  \log (\epsilon n ) )$.

\end{corollary}

The following theorem follows by an improved analysis of the dyadic trick \cite{cormode2005improved}.

\begin{theorem}\label{thm:dyadic}
There exists a data structure with space $\Oh(\frac{1}{\epsilon} \log (\epsilon n) + \log(\epsilon n) \cdot \log  ( \frac{\log (\epsilon n)}{\delta}))$ that finds the $\ell_1$ heavy hitters of $x \in \mathbb{R}^n$ in the strict turnstile model under Guarantee 3 with probability at least $1 - \delta$. The update time is $\Oh( \log^2(\frac{1}{\epsilon}) \log (\epsilon n) +  \epsilon \log(\epsilon n) \log^2(\frac{1}{\epsilon}) \log(\frac{\log( \epsilon n)}{\delta}))$ and the query time is $\Oh( \frac{1}{\epsilon}( \log^2(\frac{1}{\epsilon}) \log (\epsilon n) +  \log(\epsilon n) \log^2(\frac{1}{\epsilon}) \log (\frac{\log(\epsilon n) }{\delta})) )$.
\end{theorem}

\begin{theorem}\label{thm:deterministic}
There exists a deterministic algorithm that finds the $\epsilon$-heavy hitters of any vector $x \in \mathbb{R}^n$ using space $\Oh(k^{1+ \alpha} (\log (\frac{1}{\epsilon}) \log n)^{2+ 2/ \alpha})$. The update time is $\Oh( \mathrm{poly}(\log n))$ and the query time is $\Oh( n \cdot \mathrm{poly}( \log n))$.
\end{theorem}

\subsection{Low-failure probability algorithms in the strict turnstile model}

We present our algorithm for Theorem~\ref{thm:promise} in Algorithm~\ref{alg:count-min} below, where we set the constants $C_{R} =5, C_{\delta} = 10(\ln 4 -1), C_B = 20, C_0= 30$.

\begin{algorithm}
\caption{\textsc{Count-Min}: Scheme for Heavy Hitters in the Strict Turnstile Model}\label{alg:count-min}
\begin{algorithmic}
\State $R \gets C_{R} \log (\epsilon m) + \lceil \epsilon \log(1/\delta)/C_{\delta} \rceil$
\State $B \gets C_B k /\epsilon$
\State Pick hash functions $h_1, \ldots , h_{R}: [n] \rightarrow [B]$, each being $(C_0/\epsilon)$-wise independent
\State $C_{r,j}\gets 0$ for all $r \in [R]$ and $j \in [B]$.

\medskip

\Function{Update}{$i,\Delta$}
\State $C_{r, h_r(i)} \leftarrow C_{r,h_r(i)} + \Delta$ for all $r \in [R]$
\EndFunction

\medskip

\Function{Query}{}
\For{each $i\in P$}
\State $\hat{x}_i \gets  \min_r C_{r,h_r(i)}$
\EndFor
\State \Return the $(C_0+1)/\epsilon$ coordinates in $P$ with the largest $\hat{x}$ values.
\EndFunction
\end{algorithmic}
\end{algorithm}

\begin{proof}[Proof of Theorem~\ref{thm:promise}]
We prove first that Algorithm $1$ satisfies Guarantee $3$ with update time $\Oh(  \frac{1}{\epsilon}\log (m\epsilon) +   \log \frac{1}{\delta})$ and query time $\Oh( \frac{m}{\epsilon}\log (m \epsilon) +\log\frac{1}{\delta} )$. Later, we shall show how to modify the scheme to improve the runtime.

Let $S$ be the set of $\epsilon$-heavy hitters of $x$.  Let $T$ be any set of at most $C_0/\epsilon$ coordinates that is a subset of $P\setminus S$. Fix a hash function $h_a$. Let $B_a[T]$ denote the set of the indices of the buckets containing elements from $T$, i.e.\@ $B_a[T] = \{ b \in [B]: \exists j \in T, h_a(j) = b\}$. The probability that $|B_a[T]| <10/\epsilon $ is at most 
\[
\binom{B}{10/\epsilon}\left(\frac{10/\epsilon}{B}\right)^{\frac{C_0}{\epsilon}} \leq \left(\frac{e C_B}{10 }\right)^{\frac{10}{\epsilon}} \left(\frac{10}{C_B}\right)^{\frac{C_0}{\epsilon}}  \leq \left(\frac{e}{4}\right)^{\frac{10}{\epsilon}}.
\] 
Considering all $R$ hash functions, it follows that with probability $\geq 1 - (e/4)^{10R/\epsilon}$ there exists an index $ a_T^* \in [R]$ such that $|B_{a_T^*}[T]| \geq 10/\epsilon$. A union bound over all possible $T$ yields that with probability at least $1- \delta$, there exists such an index $a^*_T$ for each $T$ of size $C_0 /\epsilon$.

Now, let $Out$ be the set of coordinate the algorithm outputs and let $T' = Out\setminus S$. Clearly, $|T'| \geq C_0/\epsilon$. Discard some coordinates of $T'$ so that $|T'| = C_0/\epsilon$. We shall prove that there exists an element $j \in T'$ such that its estimate is strictly less than $\epsilon \|x_{-1/\epsilon }\|_1$ and hence smaller than the estimate of any element in $S$. This will imply that every element in $S$ is inside $Out$. From the previous paragraph, there exists $a_{T'}^*$ such that $|B_{a_{T'}[T']}| > 10/\epsilon$. Since we have at most $1/\epsilon$ heavy hitters and at most $1/\epsilon$ indices $b \in [B]$ such that the counter $ C_{a,b} \geq \epsilon  \|x_{-1/\epsilon} \|_1$, which implies that we have at least $10/\epsilon$ indexes (buckets) of $B$ such that the bucket has mass less than $\epsilon \|x_{1/\epsilon}\|_1$ while at least one element of $T'$ is hashed to that bucket. Therefore, the estimate for this element is less than $\epsilon \|x_{- \frac{1}{\epsilon}}\|_1$, which finishes the correctness of proof.

We note that Algorithm $1$ clearly has an update time of $\Oh( \frac{1}{\epsilon} \log(\epsilon m) + \log\frac{1}{\delta})$ and a query time $\Oh(\frac{1}{\epsilon} m \log(\epsilon m) + m\log\frac{1}{\delta})$. We now show how to improve both runtimes to achieve the advertized runtimes in the statement of the theorem.


We split our stream into intervals, which we call epochs, of length $C_0/\epsilon$: the $l$-th epoch starts from the $((l-1)C_0/\epsilon +1)$-st update and ends at the $(l C_0/\epsilon)$-th update. Let $DS$ be the data structure we constructed thus far. During an epoch we maintain a list of elements $L_0$ that were updated in this epoch. This list is initialized to the empty list when an epoch begins. When a new update $(i,\Delta)$ arrives, we store $(i,\Delta)$ to $L_0$. For the query operation, we first query $DS$ to obtain a set $L$, and then find the set $L'$ of indices $i$ such that $(i,\Delta) \in L_0$, and return $L\cup L'$ as our answer. Clearly $L \cup L'$ has at most $(2C_0 +1 )/\epsilon$ elements. When an epoch ends, we feed all updates $(i,\Delta) \in L_0$ to our data structure. Since we can obtain a $(C_0/\epsilon)$-wise independent hash function from a polynomial of degree $C_0/\epsilon$, this means that we can do $C_0/\epsilon$ evaluations in time $\Oh(\frac{1}{\epsilon} \log ^2 \frac{1}{\epsilon})$ using multipoint evaluation of polynomials
(see, e.g., Theorem 13 of \cite{kane2011fast}). Since we shall evaluate $\Oh( \log(\epsilon m) + \epsilon \log \frac{1}{\delta} )$ hash functions the amortized update time follows. A similar argument gives the query time.
\end{proof}

Observe that in the previous algorithm the analysis depends only on the set $S$ that contains the $\epsilon$-heavy hitters. By setting $P = [n]$ and taking a union bound over all possible subsets that the $\epsilon$-heavy hitters may lie in, we obtain a uniform algorithm for the strict turnstile model under Guarantee 3, namely Corollary~\ref{thm:corollary}.

We are now ready to prove Theorem \ref{thm:dyadic}, which as we said is an improved analysis of the dyadic trick \cite{cormode2005improved}, which follows from Theorem 1. We first describe the algorithm of the dyadic trick. We assume that $n$ is a power of $2$. Then for every $1 \leq l \leq \log n$, we partition $[n]$ into $2^l$ equal-sized and disjoint intervals of length $n/2^l$. Each interval will be called a node. We imagine a complete binary tree on these nodes, where there is an edge from a node/interval $u$ to a node $v$ if $v$ is an subinterval of $u$ of exactly the half length. Since in the strict turnstile model the $\ell_1$ norm of any interval/node is the sum of its elements, we can set up a Count-min sketch for every level $l$ of the tree to find out the ``heavy'' intervals at that level. If a level contains a heavy hitter then it will always have $\ell_1$ mass at least $\epsilon \|x\|_1$, while there can be at most $\frac{1}{\epsilon}$ intervals with mass more than $\epsilon \|x\|_1$. Given this observation we traverse the tree in a breadth first search fashion and at every level keep a list $L$ of all nodes that the \textsc{Count-Min} sketch on that level indicated as heavy. In the next level there will be at most $2|L|$ nodes we need to consider: just the children of nodes in $L$. At the last level, if all Count-min sketch queries succeed we shall be left with at most $\frac{1}{\epsilon}$ intervals of length $1$, that is, we shall have found all heavy hitters. The obstacle for getting suboptimal space by a factor of $\log(\frac{\log n}{\epsilon})$ stems from the fact that we have to set the parameters of \textsc{Count-Min} sketch at every level in a way such that we can afford to take a union bound over the at most $\frac{1}{\epsilon} \log n$  nodes we shall touch while traversing the tree. However, with our improved analysis of \textsc{Count-Min} sketch, we show that we can essentially avoid this additional $\log(\frac{1}{\epsilon})$ factor. Also, our algorithm with the data structure guaranteed by Theorem \ref{thm:promise} gives the stronger tail guarantee. We continue with the proof of the theorem. 

\begin{proof}[Proof of Theorem~\ref{thm:dyadic}]
Without loss of generality we can assume that $n$ is a power of $2$. We improve the analysis for levels $l\geq 1/\epsilon$. Fix such a level $l$. Then we consider the vector $ y \in \mathbb{R}^{2^l}$, the $i$-th entry of which equals 
\[y_i = \sum_{j=(i-1)\cdot \frac{n}{2^l} +1}^{i\cdot \frac{n}{2^l}} x_j\] and run point query on this vector. Every coordinate of $y$ corresponds to a node at level $l$ and since $\|x_{-\lceil \frac{1}{\epsilon} \rceil }\|_1 \geq \|y_{-\lceil \frac{1}{\epsilon} \rceil}\|_1$  finding the heavy hitters of $y$ corresponds to finding a set of nodes that contain the heavy hitters of $x$. At every level, we are solving a version of Promise Heavy Hitters with $m = \Oh(\frac{1}{\epsilon})$. We use the data structure guaranteed by Theorem \ref{thm:promise} and set the failure probability to be $\frac{\delta}{c \log(\epsilon n)}$, for some large constant $c$.   Our space consumption per level equals $\Oh(\frac{1}{\epsilon} + \log(\frac{\log(\epsilon n)}{\delta}))$. Hence, at each level we are going to find all heavy nodes with probability at least $1 - \frac{\delta}{\log (\epsilon n)}$. By a union bound over all levels the failure probability of our algorithm is at most $\delta$. This means that, while traversing the tree, we have only $\Oh(\frac{1}{\epsilon})$ candidates at each level and hence we are solving a Promise Heavy Hitters Problem with $m = \Oh(\frac{1}{\epsilon})$. At every level we need $\Oh(\frac{1}{\epsilon} + \log(\frac{\log (\epsilon n)}{\delta}))$ space and hence the total space of our algorithm is $\Oh(\frac{1}{\epsilon} \log (\epsilon n) + \log(\epsilon n) \log( \frac{\log n}{\delta}))$. Since we are considering $\Oh(\frac{1}{\epsilon} \log (\epsilon n))$ nodes in total and each function is $\Oh(\frac{1}{\epsilon})$-wise independent, similar arguments as in the proof of Theorem \ref{thm:promise} give the desired bounds.
\end{proof}

\subsection{Deterministic Algorithms in the Strict Turnstile Model}

\subsubsection{Heavy Hitters from Expanders}

In this section we show that expanders which can be stored in low space imply schemes for heavy hitters in the strict turnstile model. Then we show that the Guruswami-Umans-Vadhan expander is such an (explicit) expander.

\begin{definition}
Let $\Gamma: [N] \times [D] \rightarrow [M]$ be a bipartite graph with $N$ left vertices, $M$ right vertices and left degree $D$. Then, the graph $G$ will be called a $(k,\zeta)$ vertex expander if for all sets $ S \subseteq [N], |S| \leq k$ it holds that $\Gamma(S) \geq (1 - \zeta) |S| D$.
\end{definition}

\begin{theorem}\label{thm:expanders}
Let $\zeta,c$ be absolute constants such that $\frac{2}{1-\zeta}<c$. Suppose there exists an explicit bipartite $(\frac{c}{\epsilon},(1-\zeta)D)$ expander $\Gamma: [n] \times [D] \rightarrow [m]$. Suppose also that $\Gamma$ can be stored in space $S$, and for every $i \in [n]$ it is possible to compute the neighbours of $i$ in time $t$. Then there exists an algorithm that finds the $\epsilon$-heavy hitters of a vector $x \in \mathbb{R}^n$ with space usage $m+S$, update time $\Oh(t)$ and query time $\Oh(n \cdot t)$.
\end{theorem}

\begin{proof}
We maintain counters $C_1,\dots,C_m$. Whenever an update $(i, \Delta)$ arrives, we add $\Delta$ to all counters $C_j$ for $j$ that is a neighbour of $i$. The total update time is $\Oh(t)$. The query algorithm is exactly the same as \textsc{Count-Min}: compute $\hat{x}_i = \min_{d \in D } C_{ \Gamma(i,d)}$ and output the largest $ (c+1)/\epsilon $ coordinates. We now show the correctness of the algorithm.

Fix a vector $x \in \mathbb{R}^n$ and let $S$ be the set of its $\epsilon$-heavy hitters. Let $T$ be any other set of $c/\epsilon$ elements. Observe that for all $i \in S$ and all $j \in [m]$ such that $j,i$ are neighbours in $\Gamma$, it holds that $C_{j} \geq \epsilon\|x_{- \lceil 1/\epsilon \rceil}\|_1$. We claim that there exist adjacent $(i,j)$ for some $i \in T$ such that $ C_{j} < \epsilon \|x_{- \lceil 1/\epsilon \rceil}\|_1$, whence the theorem follows.

Suppose that the claim is false. By the expansion property of $G$, the neighbourhood of $T$ has size at least $(1-\zeta)(c/\epsilon)D$. Since the claim is false, all of these counters have value at least $\epsilon\|x_{- \lceil 1/\epsilon \rceil }\|_1$. However, the total number of counters that are at least $\epsilon\|x_{- \lceil 1/\epsilon \rceil }\|_1$, is $2 D/{\epsilon}$. It follows from our choices of $\zeta$ and $c$ that $2D/\epsilon < (1-\zeta)(c/\epsilon)D$, which is a contradiction.
\end{proof}

Next we review the Guruswami-Umans-Vadhan (GUV) expander. The construction of the GUV expander is included in Figure~\ref{fig:expanders}, and it is known that the construction does give an expander.

\begin{theorem}[{\cite{vadhan2012pseudorandomness}}]
The expander $\Gamma$ from construction $2$ is a $(h^c, q - ahc)$ expander.
\end{theorem}

The following corollary follows with appropriate instantiation of parameters.

\begin{corollary}\label{cor:expander}
For every constant $\alpha >0$ and all positive integers $N,K \leq N$ and $\epsilon$  there exists an explicit $(K,(1-\epsilon)D)$ expander with $N$ left vertices, $M$ right vertices and left-degree $D = \Oh( \log N \log K/ \epsilon)^{1+1/ \alpha}$ and $M = k^{1+\alpha} D^2$.
\end{corollary}

\begin{figure}
\begin{center}
\fbox{                                                                                                         
{\footnotesize                                                                                                 
\parbox{\textwidth} {  
\underline{Expanders from Parvaresh-Vardy Codes}:\\                                                            
\vspace{-.30in}\begin{enumerate}                                                                               
\addtolength{\itemsep}{-0.5mm}
\medskip \medskip

\item Define a neighbour function $\Gamma: \mathbb{F}_q^a \times \mathbb{F}_q \rightarrow \mathbb{F}_q \times \mathbb{F}_q^c$ by

\begin{equation} \Gamma(f,y) = [y, f_0(y), \ldots , f_{c-1}(y)], \end{equation}

where $f(Y)$ is a polynomial of degree $a-1$ over $\mathbb{F}_q$. We define $f_i(Y) = f(Y)^{h^i}~ mod~E(Y)$ where $E$ is an irreducible polynomial of degree $a$ over $\mathbb{F}_q$. 
\end{enumerate}                                                                                                
}}}
\end{center} 
\caption{Guruswami-Umans-Vadhan Expander} \label{fig:expanders}
\end{figure}
The next lemma is immediate, since the time to perform operations between polynomials of degree $a$ in $\mathbb{F}_q$ is $\Oh( c \cdot \mathrm{poly}(\log a, \log q))$. 

\begin{lemma}\label{thm:operations}
The function $\Gamma$ can be stored using $a$ words of space and the time to compute $\Gamma (f,y)$ given $f,y$ is $\Oh(c \cdot \mathrm{poly}(\log a , \log q))$.
\end{lemma}

Our main result, Theorem \ref{thm:deterministic}, then follows from combining Corollary~\ref{cor:expander}, Theorem~\ref{thm:expanders} and Lemma~\ref{thm:operations}.

\subsubsection{ Heavy hitters from Lossless Condensers}

Our starting point is an observation that in the proof above we only needed expansion on sets of size exactly $c/\epsilon$ and thus some object weaker than an expander could suffice. The right object to consider is called lossless condenser, which is essentially an expander that guarantees expansion on sets of a specific size, but not on sets of smaller size. Here we follow the definitions in \cite{vadhan2012pseudorandomness}.

In order to define a lossless condenser we need the notion of min entropy of a distribution and the notion of the total variation distance between two distributions. 

\begin{definition}[Min entropy]
Let $\mathcal{D}$ be a distribution on a finite sample space $\Omega $. The min entropy of
 $\mathcal{D}$ is defined as 
\[
H_{\infty}(\mathcal{D}) = \min_{ \omega \in \Omega}  \log\frac{1}{\Pr_{\omega \sim \mathcal{D}} [\omega]}.
\]
\end{definition}

\begin{definition}[Total variation distance]
The total variation distance between two distributions $\mathcal{P}$ and $\mathcal{Q}$ on $\Omega$ is defined to be 
\[
D_{TV}(\mathcal{P},\mathcal{Q}) = \frac{1}{2} \sum_{ \omega \in \Omega} |\mathcal{P}(\omega) - \mathcal{Q}(\omega)|.
\]
\end{definition}

We also need the definitions of a pseudorandom object called a condenser. In a nutshell, a condenser takes as input a random variable from a source which has some amount of min-entropy, and some uniform random bits. It then outputs an element following a distribution that has sufficiently large min-entropy.
Let $\mathcal{U}_n$ denote the uniform distribution on $\mathbb{F}_p^n$. 

\begin{definition}[Loseless condenser]
Let $a,b,c$ be positive integers and let $p$ be a prime number. A function $C: \mathbb{F}_p^a \times \mathbb{F}_p^b \rightarrow \mathbb{F}_p^c$ is called a lossless $(\kappa,\zeta)$ condenser if the following holds.

For every distribution $\mathcal{A}$ on $\mathbb{F}_p^b$ with $\mathcal{H}_{\infty}(\mathcal{A})  \geq \kappa$, for any random variable $A \sim \mathcal{A}$ and any `seed' variable $B \sim \mathcal{U}_b$, the distribution $(B,C(A,B))$ is $\zeta$-close to some distribution $(\mathcal{U}_b,\mathcal{Z})$ on $\mathbb{F}_p^{b+m}$ with min-entropy at least  $b+\kappa$.

\end{definition}

Equipped with these definitions, we are now ready to prove the following theorem relating lossless condensers and heavy hitters. Although we can repeat the proof of the previous section and argue that every lossless condenser is equivalent to an expander where only sets of a specific size expand, we prefer to rewrite the proof in the language of condensers.

\begin{theorem}
Let $p$ be a prime number and let $a$ be such that $n = p^a$. Let also $\kappa$ be such that $c/\epsilon = 2^{\kappa}$, where $c$ is some absolute constant. Let $Con: \mathbb{F}_p^a \times \mathbb{F}_p^b \rightarrow \mathbb{F}_p^c$ be a $(\kappa,\zeta)$ lossless condenser that can be stored in space $S$. Let $t$ be the time needed to evaluate $Con$. If $2^{-\kappa+1}/\epsilon + \zeta <1$ then there exists an algorithm that finds the $\epsilon$-heavy hitters of any $x \in \mathbb{R}^n$ with space $S + p^{c+ b}$, update time $t \cdot p^b$ and query time $\Oh(n \cdot t\cdot p^b )$.
\end{theorem}

\begin{proof}
We consider the following algorithm. We instantiate counters $C_{i,j}$ for all $i \in \mathbb{F}_p^b$ and $j \in \mathbb{F}_p^c$. Upon updating $(i,\Delta)$ we perform updates $C_{i,Con(i,j)} \leftarrow C_{i,Con(i,j)} + \Delta$ for all $j \in \mathbb{F}_p^b$. Clearly the update time is $t \cdot p^b$ and the space usage is the total number of counters plus the space needed to store the condenser, in total $\Oh(p^b \cdot p^c + S)$ words. 

The query algorithm is exactly the same as in \textsc{Count-Min}: compute $\hat{x}_i = \mathrm{min}_{j \in \mathbb{F}_p^b} C_{i,Con(i,j)}$ and output the largest $ (c+1)/\epsilon$ coordinates. We now analyse the algorithm.

Fix a vector $x \in \mathbb{R}^n$ and let $S$ be the set of its $\epsilon$-heavy hitters. Let $T$ be any other set of $c \frac{1}{\epsilon}$ elements. Observe that for all $i \in S$ and all $j \in \mathbb{F}_p^b$ it holds that $C_{i,Con(i,j)} \geq \epsilon\|x_{- \lceil 1/\epsilon \rceil}\|_1$. We claim that there exists $i \in T, j \in \mathbb{F}_p^b$ such that $ C_{i,j} > \epsilon \|x_{-\lceil 1/\epsilon \rceil }\|_1$. The theorem then follows.

Suppose that the claim is false, that is, for all $i \in T$ and all $j \in \mathbb{F}_p^b$ it holds that $C_{i,j} \geq \epsilon \|x_{-\lceil 1/\epsilon \rceil }\|_1$. Consider the uniform distribution $\mathcal{A}$ on $T$ and observe that $\mathcal{A}$ has min-entropy at least $ \kappa$. Let $A$ be a random variable drawn from $\mathcal{A}$. This implies that $(U_B,Con(A,U_B))$ is $\zeta$-close to some distribution $\mathcal{D}$ with min-entropy at least $\kappa+b$. Since the number of counters that can have a value at least $\epsilon \|x_{-\lceil 1/\epsilon \rceil}\|_1$ is $(2/\epsilon) p^b$, we have that $\Pr_{X\sim\mathcal{A}} \left\{C_{X,Con(X,U_B)} > \epsilon\|x_{-\lceil1/\epsilon \rceil }\|_1\right\} \leq (2/\epsilon) 2^b 2^{-\kappa -b} + \zeta $ . On the other hand, since $\mathcal{A}$ is supported only on elements in $S \cup T$, we have that $\mathbb{P}_{X \sim \mathcal{D}}\left\{C_{U_B,Con(A,U_B)} \geq \epsilon \|x_{-\lceil 1/\epsilon \rceil}\|_1\right\} = 1$. If $ (2/\epsilon) 2^{-\kappa} + \zeta <1$ we reach a contraction.
\end{proof}

\subsubsection{Heavy hitters from Error-Correcting List-Disjunct Matrices}

In this subsection we give another reduction to error-correcting list-disjunct matrices, a combinatorial object that appears in the context of two-stage group testing. Explicit and strongly explicit constructions of list-disjunct matrices are known \cite{ngo2011efficiently}, although they are very similar to our expander/condensers proof of the previous section and they do not yield better space complexity. We show that a sufficiently sparse error-correcting list-disjunct matrix that can be stored in low space, with the appropriate choice of parameters, can induce a scheme for heavy hitters in the strict turnstile model. Although group testing has been used in finding heavy hitters in data streams~\cite{cormode2008finding}, to the best of our knowledge this is the first time such a reduction is noticed.

Using such a matrix, one can perform two-stage combinatorial group testing when there also some false tests, either false positive or false negative. For simplicity, we give the definition only in the case of false positive tests, as we do not need more complicated ones. For an $m \times N$ binary matrix $M$ and a set $S \subseteq [n]$, let $M[S] = \{q \in [m] : \exists j \in S, M_{q,j} =1 \}$.

\begin{definition}
Let $k,\ell$ be positive integers and $e_0$ be a non-negative integer. A binary $m \times n$ matrix is called $(d,\ell,e_0)$-list-disjunct if for any disjoint sets $S,T \subseteq [n]$ with $|S|= k, |T| = \ell$, the following holds: in any arbitrary subset $X$ of $M[T] \setminus M[S]$ of size at most $e_0$, there exists a column $j^* \in T$ such that $M[\{j\}] \setminus (X \cup M[S]) \neq \emptyset$.
\end{definition}

\begin{theorem}
Let $s,n$ be integers with $s \leq n$.
Suppose there exists an $m \times n$ matrix $M$ that is $(1/\epsilon , c/\epsilon , s/\epsilon) $ list-disjunct, where $s$ is some integer and $C_0$ some absolute constant. Suppose also that:
\begin{itemize}
\item $M$ can be stored in space $S$;
\item The column sparsity of $M$ is $s$;
\item Each entry of $M$ can be computed in time $t$.
\end{itemize}
Then there exists a streaming algorithm which finds the $\epsilon$-heavy hitters of any vector $x$ using space $S + m$, having update time $s \cdot t$ and query time $\Oh(n \cdot st)$.
\end{theorem}

\begin{proof}
We use $M$ as our sketch, i.e., we have access to $y = Mx$. The inner product of each row of matrix $M$ with $x$ defines a counter.
We describe the query algorithm. For each coordinate $i$, compute $\hat{x}_i$ as the minimum value over all counters it participates in. That is, $\hat{x}_i = \mathrm{min}_{q: M_{q,i} =1 } y_q$. The query algorithm outputs the list $L$ containing the coordinates with the largest $(c+1)/\epsilon$ coordinates in $\hat{x}$. 

Next we show the correctness of the query algorithm. Let $S$ be the set of $\epsilon$-heavy hitters and let $R$ be the set of coordinates that the algorithm outputs. Define $T = R\setminus S$. For any $i \in S$, we know that $\hat{x}_i \geq \epsilon \|x_{-\lceil 1/\epsilon \rceil }\|_1$. Let $X = \{ q : y_q > \epsilon \|x_{- \lceil 1/\epsilon \rceil}\|_1~\mathrm{and}~ \nexists i\in S:M_{q,i}=1 \}$, that is, the set of rows which appear to be `heavy' but contain no $\epsilon$-heavy hitter. We claim that $|X| \leq s /\epsilon$. Indeed, if $|X| > s/\epsilon$,  the total $\ell_1$ mass of these counters would be more than $s \|x_{-\lceil 1/\epsilon \rceil}\|_1$. But since every coordinate $i$ participates in at most $s$ counters, the total $\ell_1$ mass of counters with no heavy hitters is at most $s \|x_{-\lceil 1/\epsilon \rceil }\|_1$. This gives us the desired contradiction. 

Assume now that there exists $i \in S$ which is not included in $L$. This means that for every $j \in T$, $\hat{x}_j \geq \epsilon\|x_{-\lceil 1/\epsilon \rceil }\|_1$, which means that $M[j] \subseteq M[S] \cup X$. Define sets $S',T'$ by moving some coordinates from $T$ to $S$ if needed so that $|T'| = c /\epsilon, |S'| = 1/\epsilon$. Observe that it still holds that $M[j] \subseteq M[S] \cup X$ for all $j \in T'$. But this violates the definition of $(1/\epsilon,c/\epsilon,s/\epsilon)$ list-disjunct matrix. This proves correctness of our algorithm.

Clearly, the space that the streaming algorithms uses is $S+m$, the space needed to store $A$ plus the space needed to store $y$. Moreover, the update time is $s \cdot t$.
\end{proof}

\section{Lower Bounds}
\subsection{Preliminaries}\label{sec:lb_prelim}
We start with some standard facts about the Gaussian distribution. 

\begin{fact}[Total Variation Distance Between Gaussians]\label{lem:gtvd}
  Let $N(\mu,I_r)$ denotes the $r$-dimensional Gaussian distribution with mean $\mu$ and identity covariance, then
  \[
  D_{TV}(N(0, I_r), N(\tau, I_r)) = \Pr_{g\sim N(0,1)}\left\{|g| \leq \frac{\|\tau\|_2}{2}\right\}.
  \]
\end{fact}

\begin{proof}
  Let $U$ be an $r\times r$ orthogonal matrix which rotates $\tau$ to $\|\tau\|_2 \cdot e_1$, where $e_1$ is the first standard unit vector. Then by rotational invariance,
  $D_{TV}(N(0,I_n), N(\tau,I_n)) = D_{TV}(N(0,1), N(\|\tau\|_2, 1))$. It then follows from ~\cite[Section 3]{tvd} that
  \[
  D_{TV}(N(0,1), N(\|\tau\|_2, 1)) = \Pr_{g\sim N(0,1)}\left\{|g| \leq \frac{\|\tau\|_2}{2}\right\}.\qedhere
  \]
\end{proof}

\begin{theorem}[Lipschitz concentration, {\cite[p21]{LT91}}]\label{thm:concentrate}
Suppose that $f:\mathbb{R}^n \rightarrow \mathbb{R}$ is $L$-Lipschitz with respect to the Euclidean norm, i.e.\@ $|f(x)-f(y)| \leq L\|x-y\|_2$ for all $x, y \in \mathbb{R}^n$.
Let $x \sim N(0, I_n)$ and let $f$ be $L$-Lipschitz with respect to the Euclidean norm. Then
    \[
    \Pr\left\{|f(x) - \E f(x)| \geq t\right\} \leq 2 e^{-\frac{t^2}{2L^2}},\quad t\geq 0.
    \]
  \end{theorem}

\begin{fact}[Concentration of $\ell_1$-Norm]\label{fact:smallNorm}
  Suppose $x \sim N(0, I_n)$ and $n \geq 32\ln(6/\delta)$. 
  Then
  $$\Pr \left\{\frac{n}{8} \leq \|x\|_1 \leq \frac{3n}{4} \right\} \geq 1-\frac{\delta}{3}.$$
  \end{fact}
\begin{proof}
The function $f(x) = \|x\|_1$ satisfies $|f(x)-f(y)| \leq \|x-y\|_1 \leq \sqrt{n} \|x-y\|_2$ for all $x,y$ in $\mathbb{R}^n$ and so by Theorem \ref{thm:concentrate}
we have 
\[
\Pr_{x\sim N(0,I_n)}\{|\|x\|_1 - \E\|x\|_1| > t\} \le 2e^{-\frac{t^2}{2n}}
\]
Note that since $g \sim N(0,1)$, $\E|g| = \sqrt{2/\pi}$, and so we have $\E\|x\|_1 = \sqrt{2/\pi}$. Hence, setting $t = n/4$, we have that $n/8 \leq \|x\|_1 \leq 3n/4$ with probability at least
$1-2e^{-n/32}$, which is at least $1-\delta/3$ provided that $n \geq 32\ln(6/\delta)$.
\end{proof}

\begin{fact}[Concentration of $\ell_2$-Norm]\label{fact:small2Norm}
  Suppose $x \sim N(0, I_n)$ and $n \geq 18\ln(6/\delta)$. Then
  $$\Pr \left\{\frac{\sqrt{n}}{2} \leq \|x\|_2 \leq \frac{3\sqrt{n}}{2} \right\} \geq 1-\frac{\delta}{3}.$$ 
\end{fact}
\begin{proof} The function $f(x) = \|x\|_2$ satisfies $|f(x) - f(y)| \leq \|x-y\|_2$ for all $x, y$ in $\mathbb{R}^n$ and so by
  Theorem \ref{thm:concentrate} we have
  \begin{eqnarray*}
\Pr\{\|x\|_2 - \E\|x\|_2| > t\} \leq 2e^{-t^2/2}
  \end{eqnarray*}
  if $x \sim N(0, I_n)$. It is well-known (see, e.g., \cite{crpw10})
  that $\frac{n}{\sqrt{n+1}} \leq \E\|x\|_2 \leq \sqrt{n}$. Hence, setting
  $t = \sqrt{n}/3$, we have that $\sqrt{n}/2 \leq \|x\|_2 \leq 3\sqrt{n}/2$ with probability at least $1-2e^{-n/18}$,
  which is at least $1-\delta/3$ provided that $n \geq 18\ln(6/\delta)$.
\end{proof}

  \begin{fact}[Univariate Tail Bound]\label{fact:coordinate}
    Let $g\sim N(0,1)$. There exists $\delta_0 > 0$ such that it holds for all $\delta < \delta_0$ that 
    $\Pr\{|g| \leq 4 \sqrt{\log(1/\delta)}\} \geq 1-\delta/3$. 
    \end{fact}
\begin{proof}
It is well-known that (see, e.g.~\cite{gordon1941})
\[
\Pr\{g \geq x\} \leq \frac{e^{-x^2/2}}{x \sqrt{2\pi}},
\]
and so
$\Pr\{g \geq 2 \sqrt{\log(1/\delta^2)}\} \leq \frac{\delta^2}{2 \sqrt{2\pi \log(1/\delta)}} \leq \frac{\delta}{6},$
for $\delta$ less than a sufficiently small absolute constant $\delta_0 > 0$. Hence, by symmetry of the normal
distribution, $|g| \leq 4 \sqrt{\log(1/\delta)}$ with probability at least $1-\delta/3$.
\end{proof}
  
In our proofs we are interested in lower bounding the number $r$ of rows of a sketching matrix $S$.
We can assume throughout w.l.o.g. $S$ has orthonormal rows, since given $Sx$, one can compute $RSx$ for any $r \times r$ change-of-basis matrix $R$.

\subsection{Deterministic Lower Bounds for $\ell_1$-Heavy Hitters}
\begin{theorem}[Strict turnstile deterministic lower bound for Guarantees 1,2]\label{thm:det1}
  Assume that $n = \Omega(\epsilon^{-2})$. Any sketching matrix $S$ must have $\Omega(\epsilon^{-2})$ rows if, in the strict turnstile model,
  it is always possible to recover from $Sx$ a set which contains all the $\epsilon$-heavy hitters of $x$ and contains no items which are not $(\epsilon/2)$-heavy hitters. 
  \end{theorem}
\begin{proof}
  Let $r$ be the number of rows of the sketching matrix $S$, which, w.l.o.g., has orthonormal rows. Because the columns of $S^T$ are orthonormal,
  $\sum_{i=1}^n \|S^TSe_i\|_2^2 = \sum_{i=1}^n \|Se_i\|_2^2 = r$ and
  $$\sum_{i=1}^n \|S^TSe_i\|_1 \leq \sqrt{n} \sum_{i=1}^n \|S^TSe_i\|_2 \leq n \left (\sum_{i=1}^n \|S^TSe_i\|_2^2 \right )^{1/2} \leq n \sqrt{r},$$
  where the first inequality uses the relationship between the $\ell_1$-norm and the $\ell_2$-norm, and the second inequality is the Cauchy-Schwarz inequality. 
  It follows by an
  averaging argument that there exists an $i^* \in [n]$ for which $\|S^TSe_{i^*}\|_2^2 \leq 2r/n$ and $\|S^TSe_{i^*}\|_1 \leq 2\sqrt{r}$. Consider the possibly negative vector $v = e_{i^*} - S^TSe_{i^*}$, which is in the kernel of $S$, since $I -S^TS$ projects onto the space orthogonal to the row space of $S$.
  Then $v_{i^*} \geq 1 - 2r/n \geq 1/2$, for $r \leq n/4$. Also $\|v\|_1 \leq 1 + \|S^TSe_{i^*}\|_1 \leq 2\sqrt{r} + 1$. 

Define $w$ to be the vector with $w_j = |v_j|$ for all $j \neq i^*$, and $w_{i^*} = 0$. Note that $w$ is non-negative. Then $\|w+v\|_1 \leq 4\sqrt{r}+2$, while $(w+v)_{i^*} \geq 1/2$, and note that $w+v$ is also a non-negative vector. Since both $w$ and $w+v$ are non-negative vectors, either can be presented to $S$ in the strict turnstile model. However, $S(w+v) = Sw$, and so any algorithm cannot distinguish input $w+v$ from input $w$. However, $w_{i^*} = 0$ while $(w+v)_{i^*} \geq 1/2$. For the algorithm to be correct, $i^*$ cannot be an $(\epsilon/2)$-heavy hitter for either vector, which implies $\frac{(1/2)}{4 \sqrt{r} + 2} \leq \frac{\epsilon}{2}$, that is, $r = \Omega(1/\epsilon^2)$. 
\end{proof}

\begin{theorem}[Turnstile deterministic lower bound for Guarantee 3]\label{thm:det2} 
 Assume that $n = \Omega(\epsilon^{-2})$. Any sketching matrix $S$ must have $\Omega(\epsilon^{-2})$ rows if, in the turnstile model, some algorithm never fails in returning a superset of size $O(1/\epsilon)$ containing the $\epsilon$-heavy hitters. Note that it need not return approximations to the values of the items in the set which it returns.
\end{theorem}
\begin{proof}
  Let $r$ be the number of rows of the sketching matrix $S$, which, w.l.o.g., has orthonormal rows. As in the proof of Theorem \ref{thm:det1},
  $\sum_{i=1}^n \|S^TSe_i\|_2^2 = r$ and $\sum_{i=1}^n \|S^TSe_i\|_1 \leq n \sqrt{r}$. 
  It follows that by an averaging argument there is a set $T$ of $9n/10$ indices $i \in [n]$ for which
  both $\|S^TSe_i\|_2^2 \leq 20r/n$ and $\|S^TSe_i\|_1 \leq 20\sqrt{r}$.
  Let $v_i = e_i - S^TSe_i$, which is in the kernel of $S$ for each $i$.
  Furthermore, as in the proof of Theorem \ref{thm:det1},
  for $i \in T$, the $i$-th coordinate $(v_i)_i$ of $v_i$, is at least $1/2$ for $r \leq n/4$, and also $\|v_i\|_1 \leq 2 \sqrt{r}+1$. 

  Now, since $Sv_i = 0$ for all $v_i$, a turnstile streaming algorithm cannot distinguish any of the input vectors $v_i$ from the input $0$.
  The output of the algorithm on input $0$ can contain at most $O(1/\epsilon)$ indices.
  But $|T| \geq 9n/10 = \omega(1/\epsilon)$, so for at least one index $i \in T$, the algorithm will be wrong if $i$ is an $\epsilon$-heavy hitter for $v_i$,
  so we require $\frac{(1/2)}{2\sqrt{r}+1} \leq \epsilon$, that is, $r = \Omega(1/\epsilon^2)$. 
  \end{proof}

\subsection{Randomized Lower Bounds}
\subsubsection{$\ell_1$-Heavy Hitters}
We assume in this section that $1/\epsilon \geq C \sqrt{\log(1/\delta)}$ for a sufficiently large constant $C > 0$. To obtain
an $\Omega(\epsilon^{-1} \sqrt{\log(1/\delta)}$ lower bound on the number of rows of the sketching matrix, such an assumption is necessary,
otherwise for small enough $\delta$ (as a function of $n$ and $\epsilon$)
the lower bound would contradict the $\epsilon^{-2} \poly(\log n)$ deterministic upper bounds of \cite{gm06,nnw12}.


\begin{theorem}(Randomized Turnstile $\ell_1$-Heavy Hitters Lower Bound for Guarantees $1,2$)
  Assume that
  $1/\epsilon \geq C \sqrt{\log(1/\delta)}$ and suppose $n \geq  \left \lceil 64 \epsilon^{-1} \sqrt{\log(1/\delta)} \right \rceil$. Then
    for any sketching matrix $S$, it must have $\Omega(\epsilon^{-1} \sqrt{\log(1/\delta)})$ rows if, in the turnstile model, it succeeds
    with probability at least $1-\delta$ in returning a set containing all the $\epsilon$ $\ell_1$-heavy hitters and containing no
    items which are not $(\epsilon/2)$ $\ell_1$-heavy hitters. 
\end{theorem}
\begin{proof}
  Let the universe size be $n = \left \lceil 64 \epsilon^{-1} \sqrt{\log(1/\delta)} \right \rceil$
  (note if the actual universe size is larger, we can set all but
the first $\left \lceil 64\epsilon^{-1} \sqrt{\log(1/\delta)} \right \rceil$ coordinates of our input to $0$). 
  Note that the assumption
in the previous paragraph implies that $n \geq 32\ln(6/\delta)$, which we need in order to apply Fact~\ref{fact:smallNorm}.

Let $r$ be the number of rows of the sketch matrix $S$, where $r \le n$. Note that if $r > n$, then we immediately
obtain the claimed $\Omega(\epsilon^{-1} \sqrt{\log(1/\delta)})$ lower bound.

\paragraph{Hard Distribution.}
Let $I$ be a uniformly random element in $[n]$.
\\\\
{\bf Case 1:}
Let $\eta$ be the distribution $N(0, I_n)$, and suppose $x \sim \eta$. By Fact \ref{fact:smallNorm},
$\|x\|_1 \geq n/8 \geq 8 \epsilon^{-1} \sqrt{\log(1/\delta)}$ with probability $1-\delta/3$.
By Fact \ref{fact:coordinate}, $|x_I| \leq 4 \sqrt{\log(1/\delta)}$ with probability $1-\delta/3$.

Let $\mathcal{E}$
be the joint occurrence of these events, so that $\Pr[\mathcal{E}] \geq 1-2\delta/3$. 

By our choice of $n$, it follows that if $\mathcal{E}$ occurs, then $|x_I| \leq \frac{\epsilon}{2}\|x\|_1$, and therefore $I$ cannot be output by
an $\ell_1$-heavy hitters algorithm if the algorithm succeeds. 
\\\\
{\bf Case 2:} Let $\gamma$ be the distribution $(4 \epsilon n) e_I + N(0, I_n)$, where $I$ is drawn uniformly at random from $[n]$,
and $e_I$ denotes the standard unit vector in the $I$-th direction. Suppose $x \sim \gamma$, and let $y \sim N(0, I_n)$.
By Fact \ref{fact:smallNorm}, $\|y\|_1 \leq \frac{3n}{4}$ with probability $1-\delta/3$.
By Fact \ref{fact:coordinate}, $|x_I| \leq 4 \sqrt{\log(1/\delta)}$ with probability $1-\delta/3$.

Let $\mathcal{F}$
be the joint occurrence of these events, so that $\Pr[\mathcal{F}] \geq 1-2\delta/3$. 

Note that $4 \sqrt{\log(1/\delta)} \leq \epsilon n$, and consequently if $\mathcal{F}$ occurs, 
then $|x_I| \geq (4 \epsilon n) - 4 \sqrt{\log(1/\delta)} \geq 3 \epsilon n$, while 
$\|x\|_1 \leq 4\epsilon n + \|y\|_1 \leq n$, provided $\epsilon \leq 1/16$.
Consequently, if $\mathcal{F}$ occurs,
for an $\ell_1$-heavy hitters algorithm to be correct, it must output $I$. 
%

\paragraph{A Conditioning Argument.}
We let $\eta'$ be the distribution of $\eta$ conditioned on $\mathcal{E}$, and let $\gamma'$ be the distribution of $\beta$ conditioned on $\mathcal{F}$.
For a distribution $\mu$ on inputs $y$, we let $\bar{\mu}$ be the distribution of $Sy$.

Note that any $\ell_1$-heavy hitters algorithm which succeeds with probability at least $1-\delta$ can decide,
with probability at least $1-\delta$, if $x \sim \eta'$ or if $x \sim \gamma'$. Consequently, $D_{TV}(\bar{\eta'}, \bar{\gamma'}) \geq 1-\delta$.

On the other hand, we have
\begin{eqnarray*}
  D_{TV}(\bar{\eta}, \bar{\gamma}) & \geq & D_{TV}(\bar{\eta'}, \bar{\gamma'}) - D_{TV}(\bar{\eta'}, \bar{\eta}) - D_{TV}(\bar{\gamma'}, \bar{\gamma})\\
  & \geq & D_{TV}(\bar{\eta'}, \bar{\gamma'}) - \frac{2\delta}{3} - \frac{2\delta}{3}\\
  & \geq & 1-\frac{7\delta}{3},
  \end{eqnarray*}
where the first inequality is the triangle inequality. The second inequality follows from the fact that for a distribution $\mu$ and event $\mathcal{W}$,
$D_{TV}(\mu, \mu \mid \mathcal{W}) = \Pr[\not \mathcal{W}]$, together with our bounds on $\Pr[\mathcal{E}]$ and $\Pr[\mathcal{F}]$ above. 

Therefore, to obtain our lower bound, it suffices to show if the number $r$ of rows of $S$ is too small, then it cannot hold that
$D_{TV}(\bar{\eta}, \bar{\gamma}) \geq 1-\frac{7\delta}{3}$. 

\paragraph{Bounding the Variation Distance.}
Since $S$ has orthonormal rows, by rotational invariance of the Gaussian distribution, the distribution of $\bar{\eta}$
is equal to $N(0, I_r)$. Also by rotational invariance of the Gaussian distribution, the distribution
of $\bar{\gamma}$ is $(4 \epsilon n) S_I + N(0, I_r)$, where $S_I$ is the $I$-th column of $S$. 

Since $S$ has orthonormal rows,
by a Markov bound, for $9/10$ fraction of values of $I$, it holds that $\|S_I\|_2^2 \le \frac{10r}{n}$.
Call this set of columns $T$. 

Let $\mathcal{G}$ be the event that $I \in T$. Then $\Pr[\mathcal{G}] \geq 9/10$.
Suppose $D_{TV}(\bar{\eta}, \bar{\gamma}) \geq 1-\frac{7\delta}{3}$. We can write
$\bar{\gamma} = \Pr[\mathcal{G}] \cdot (\bar{\gamma} \mid \mathcal{G}) + \Pr[\neg \mathcal{G}] (\bar{\gamma} \mid \neg \mathcal{G})$, and so if the number $r$
of rows of $S$ is large enough so that $D_{TV}(\bar{\eta}, \bar{\gamma}) \geq 1-\frac{7\delta}{3}$, then 
\begin{eqnarray*}
  1 - \frac{7\delta}{3} & \leq & D_{TV}(\bar{\eta}, \bar{\gamma})\\
  & = & \frac{1}{2}\|\bar{\eta} - \bar{\gamma}\|_1 \\
  & = & \frac{1}{2} \|\bar{\eta} - \Pr[\mathcal{G} \cdot (\bar{\gamma} \mid \mathcal{G}) - \Pr[\neg \mathcal{G}] (\bar{\gamma} \mid \neg \mathcal{G}) \|_1\\
    & = & \frac{1}{2} \|\Pr[\mathcal{G} (\bar{\eta} - (\bar{\gamma} \mid \mathcal{G})) + \Pr[\neg \mathcal{G}] (\bar{\eta} - (\bar{\gamma} \mid \neg \mathcal{G}))\|_1\\
      & \leq & \frac{1}{2} \Pr[\mathcal{G} \|\bar{\eta} - \bar{\gamma} \mid \mathcal{G}\|_1 + \frac{1}{2} (1-\Pr[\mathcal{G}]) \|\bar{\eta} - \bar{\gamma} \mid \neg \mathcal{G}\|_1\\
      & \leq & \frac{1}{2} \left (\frac{9}{10} \|\bar{\eta} - \bar{\gamma} \mid \mathcal{G}\|_1 + \frac{1}{10} \cdot 1 \right ),
\end{eqnarray*}
and so $D_{TV}(\bar{\eta}, \bar{\gamma} \mid \mathcal{G}) \geq 1 - \frac{70\delta}{27}$. 

The variation distance between $N(0, I_r)$ and $(4 \epsilon n)S_i + N(0, I_r)$ for a fixed $i \in T$ is, by rotational invariance and by
rotating $S_i$ to be in the same direction as the first standard unit vector $e_1$, the same as the variation distance between $N(0, I_r)$ and 
$4 \epsilon n \|S_i\|_2 e_1 + N(0, I_r),$
which is equal to the variation distance between $N(0, 1)$ and $N(4 \epsilon n \|S_i\|_2, 1)$.
Using that $i \in T$ and so $\|S_i\|_2 \leq \sqrt{10r/n}$, we apply Lemma \ref{lem:gtvd} to obtain that
the variation distance is at most $\Pr_{g\sim N(0,1)}\{|g| \leq 2 \epsilon n \sqrt{10r/n}\}$. It follows that
\[
  D_{TV}(\bar{\eta}, \bar{\gamma} \mid \mathcal{G}) \leq \sum_{i \in T} \frac{1}{|T|} D_{TV}(\bar{\eta}, \bar{\gamma} \mid (I = i))
  \leq \Pr_{g\sim N(0,1)}\left\{|g| \leq 2 \epsilon n \cdot \sqrt{\frac{10r}{n}}\right\}.
\]
Supposing the number $r$ of rows of $S$ is large enough so that $D_{TV}(\bar{\eta}, \bar{\gamma} \mid \mathcal{G}) \geq 1- \frac{70\delta}{27}$, this implies
$\Pr\{|g| \leq 2 \epsilon n \sqrt{10r/n}\} \geq 1- \frac{70 \delta}{27}$, or equivalently
\begin{equation}\label{eqn:mainBound}
  \Pr_{g\sim N(0,1)}\left\{|g| > 2 \epsilon n \sqrt{\frac{10r}{n}}\right\} \leq \frac{70 \delta}{27}.
  \end{equation}

Suppose that $r \leq \alpha \epsilon^{-1} \sqrt{\log(1/\delta)}$. Then
$2 \epsilon n \sqrt{10r/n} = 2 \epsilon \sqrt{10rn} \leq C'\sqrt{\alpha \log(1/\delta)}$ for some absolute constant $C'$. Take $\alpha = 1/(2C'^2)$. Invoking the well-known bound (see, e.g.\@ \cite{gordon1941})
\begin{eqnarray*}
\Pr_{g\sim N(0,1)}\{g \geq x\} \geq \frac{e^{-x^2/2}}{\sqrt{2\pi}} \left (\frac{1}{x+1/x}\right ),
\end{eqnarray*}
we have that $\Pr_{g\sim N(0,1)}\{|g| > 2 \epsilon n\sqrt{10r/n}\} \geq C''\sqrt{\delta/\log(1/\delta)}$, when $\delta$ is small enough, contradicting that this probability needs to be
at most $70 \delta/27$ by (\ref{eqn:mainBound}).

It must therefore hold that $r > \alpha\epsilon^{-1} \sqrt{\log(1/\delta)}$. This completes the proof. 
\end{proof}

\section{Non-adaptive Sparse Recovery}\label{sec:non-adaptive}

\subsection{Weak System}
We start with the following lemma.

\begin{lemma}\label{lem:weak_identification}
Let $\delta, \alpha > 0$. There exist absolute constants $c_1,c_B>0$ and a randomized construction of a $(k,c_1)$-weak identification matrix with the $(c_B\frac{k}{\epsilon} , \delta,\epsilon)$ guarantee; the matrix has $\Oh( (\frac{k}{\epsilon} + \frac{1}{\epsilon}\log\frac{1}{\delta}) \log\frac{\epsilon n}{k}))$ rows and the identification time is $\Oh( (\frac{k}{\epsilon} + \frac{1}{\epsilon}\log \frac{1}{\delta}) \log^{2+\gamma}\frac{\epsilon n}{k} )$ (where $\gamma>0$ is an arbitrarily small constant) when $k/\epsilon\leq n^{1-\alpha}$ and $\Oh(n\log^2 n)$ when $k/\epsilon\geq n^{1-\alpha}$.
\end{lemma}

\begin{proof}
We describe the matrix as a scheme which performs a set of linear measurements on the vector $x$. Without loss of generality we can assume that $H_{k,\epsilon}(x)$ has size $k$, otherwise we can expand it to size $k$ without affecting the guarantee of the weak identification matrix.

Let $R=c_R\lceil\frac{1}{k}\log(\frac{1}{\delta})\rceil$ and $B=c_B k/\epsilon$. For each $r\in [R]$, we hash every coordinate $i \in [n]$ to $B$ buckets using a $B$-wise independent hash function $h_r :[n] \to [B]$. Let $N_j^r$ be binary variables defined as follows: $N_j^{r}=1$ iff $|\{i: h_r(i)= j \}| \geq C_0 n/B$. We call the $j$-bucket in the $r$-th repetition big if $N_j^r=1$ and small otherwise. It is clear that 
at most $B/C_0$ buckets will be big. Relabel the buckets so that the first $(1- 1/C_0) B$ buckets are small buckets. It follows from a Chernoff bound that with probability at least $1 - e^{-\Omega(k)}$ all but $c_0k$ coordinates will land in small buckets, for some constant $c_0 < 1$. Let this set be $\mathcal{T}$, and define $\hat{H}_{k,\epsilon}(x) = H_{k,\epsilon}(x) \setminus \mathcal{T}$. We now show that all but a small constant fractions of elements in $\hat{H}_{k,\epsilon}(x)$ will be isolated in the small buckets, and, furthermore, with low $\ell_2$ noise from the tail. 

Define 
\begin{gather*}
T_{\mathrm{good}}^{(r)} = \left\{ b \in \left[(1-1/C_0)B\right]: \text{there exists exactly one }j \in \hat{H}_{k,\epsilon}(x) \text{ such that }h_r(j) = b \right\},\\
T_{\mathrm{bad}}^{(r)} = \left\{  b \in \left[(1-1/C_0)B\right]: \sum_{\substack{i \notin \hat{H}_{k,\epsilon}(x)'\\ h_r(i)=b}} x_i^2 \geq \frac{\epsilon}{5k} \|x_{-k}\|_2^2\right\}.
\end{gather*}

In other words, $T_{\mathrm{good}}^{(r)}$ is the set of the buckets in iteration $r$ which receive exactly one element of $\hat{H}_{k, \epsilon}(x)$, while $T_{\mathrm{bad}}^{(r)}$ is the set of buckets which receive energy more than $(\epsilon/5k) \|x_{-k}\|_2^2$ from elements outside of $\hat{H}_{k,\epsilon}(x)$ (and hence outside of $H_{k,\epsilon}(x)$).

Since every two buckets in $T_{\mathrm{bad}}^{(r)}$ share no coordinates outside $\hat{H}_{k,\epsilon}(x)$, it must hold that $|T_{\mathrm{bad}}^{(r)}| \leq 5k/\epsilon$. The probability that there exist more than $c_3 k$ coordinates in $\hat{H}_{k,\epsilon}(x)$ that land in a bucket in $T_{\mathrm{bad}}^{(r)}$ is at most ${k \choose c_3 k} (\frac{C_0}{c_B})^{ c_3 k} \leq e^{-\Omega(k)}$, by choosing $c_B$ large enough. On the other hand, it is a standard result (see, e.g.,~\cite[Lemma 23]{GSTV06}) that at least $(1-c_4)k$ elements in $\hat{H}_{k,\epsilon}(x)$ are isolated with probability at least $1-e^{-\Omega(k)}$, and thus $|T_{\textrm{good}}^{(r)}|\geq (1-c_4)k$. 

Overall with probability $\geq 1-e^{-\Omega(k)}$, at least $(1-c_3-c_4)k$ elements in $\hat{H}_{k,\epsilon}(x)$ land in buckets in $T_{\textrm{good}}^{(r)}\setminus T_{\mathrm{bad}}^{(r)}$. Each bucket in $T_{\textrm{good}}^{(r)}\setminus T_{\mathrm{bad}}^{(r)}$ is a $1$-heavy hitter problem with signal length at most $C_0n/B$ and the energy of the heavy hitter is at least $5$ times the noise energy. 
We now invoke Theorem \ref{thm:btree} with different parameters depending on the value of $k/\epsilon$:
\begin{itemize} 
\item If $k/\epsilon \leq n^{1-\alpha}$ we use $\Oh(\log n)$ measurements by setting the universe to be $[n]$. Then we obtain success probability at least $1 - \frac{1}{\poly(\log n)}$ and decoding time $\Oh(\log^{2+\gamma} n)$ for any constant $\gamma>0$. \item If $ k/\epsilon > n^{1-\alpha}$ we use $\Oh(\log(\epsilon n/ k))$ measurements, by setting the universe size to be of size $\Theta(\epsilon n/ k)$. 
\end{itemize}
We proceed by collecting the coordinate found in each bucket, 
and keep the coordinates with at least $R/2$ occurrences. Since there are at most $BR$ buckets, the size of the returned coordinates is at most $2B = 2c_Bk/\epsilon$.
We note though that in the former case ($k\leq n^{1-\alpha}$), since we assume the universe size is $n$ so we can recover the coordinate indices directly, while in the latter case ($k\geq n^{1-\alpha}$), we need to invert the hash function $h_r$ to obtain the indices of the coordinates. For that, we iterate over all $i \in [n]$, and calculate all values $h_r(i)$. This can be done in time $O(n \log^2 (k/\epsilon))$, by splitting the interval $[n]$ into $n\epsilon /k$ blocks, and performing fast-multipoint evaluation in all every interval, for a cost of $\Oh((k/\epsilon) \log^2 (k/\epsilon))$ per interval. At the end, all values with the same $h_r$ value are grouped, and we perform QuickSelect in every group to find the index of the corresponding element.

Since the failure probability of the $b$-tree is at most a constant, the probability that more than $c_5 k$ buckets fail is at most $e^{ -\Omega(k)}$. By considering all $R$ iterations, this means that with probability $1 - \exp(-\Omega(Rk)) \geq 1- \delta$ at most $ c_5(c_3+c_4)k$ elements of $\hat{H}_{k,\epsilon}(x)$ will not be recognized, and hence at most $c_0k + c_5(c_3+c_4)k = c_1 k$ of $H_{k,\epsilon}(x)$ will not be recognized. 
\end{proof}

\begin{lemma}\label{lem:estimation}
Let $T\subseteq [n]$ be a set of indices such that $|T| \leq c_T k/\epsilon$ for some absolute constant $c_T$ and $\zeta\leq \frac{1}{2}$ be an absolute constant. The \textsc{Count-Sketch} scheme of $B = c_B(c_T+1)k/\epsilon$ buckets and $R = c_R(\log\frac{1}{\epsilon}+\frac{1}{k}\log\frac{1}{\delta})$ repetitions yields an estimate $\hat x$ such that 
\[
\left| \left\{i\in T: |x_i - \hat x_i|^2 > \frac{\epsilon}{16k}\|x_{-k}\|_2^2 \right\} \right| \leq \zeta k,\] 
with probability $\geq 1-\delta$. 
\end{lemma}
\begin{proof}
For the purpose of analysis only we may assume that $H_k(x)\subseteq T$, otherwise we can include $H_k(x)$ in $T$ and replace $c_T$ with $c_T+1$. We call elements in $T$ \emph{candidates}, and the elements not in $T$ noise elements.

Consider hashing all $n$ elements into $B$ buckets, using fully independent hash functions. We also combine with fully independent random signs.
We say a bucket $b$ in repetition $r$ is good if 
\[
\sum_{\substack{i\not\in T\\ h_r(i)=b}} x_i^2 \leq \frac{\epsilon}{160k}\left\|x_{T^c}\right\|_2^2,
\]
otherwise we say that bucket $b$ in repetition $r$ is bad. 

First we claim that, with probability at least $1-\exp(-\Omega(|T|R))\geq 1-\epsilon\delta^{1/\epsilon}$, at least $(1-\theta_1)T$ candidates are isolated and land in good buckets in at least $(1-\theta_2)R$ repetitions. 
The isolation claim is essentially the same to those in~\cite{GSTV06,PS12,GLPSl1}, nevertheless we give a proof below for
completeness. In each repetition $r$, it follows from a standard result (see~\cite[Section 4.3]{GG11}) that at least $(1-\theta_3)|T|$ candidates are isolated with probability at least $1-\exp(-c_1\theta_4|T|)$, by choosing $c_B$ large enough. 
The other claim of landing in good buckets follows from a standard argument. In each repetition there are at most $\theta_5 B$ 
buckets that are bad, since at most $\theta_5 B$ buckets contain noise energy greater than $\|x_{T^c}\|_2^2/(\theta_5 B) \leq \frac{\epsilon}{160 k}\|x_{T^c}\|_2^2$. A standard application of the Chernoff bound shows that with probability at least $1-\exp(-\Omega(|T|))$, at least $(1-\theta_6)|T|$ candidates are hashed into good buckets. Call a pair (candidate, repetition) good if
the candidate is isolated and lands in a good bucket in that repetition. 
Therefore in each repetition with probability at least $1-\exp(-\Omega(|T|))$ there are at least $(1-\theta_7)|T|$ good pairs.
Taking a Chernoff bound over $R$ repetitions, with probability $\geq 1-\exp(-\Omega(|T|R))$, at least $(1-\theta_8)R$
repetitions contain $(1-\theta_7)|T|$ good pairs. Conditioned on this event, we know at least $(1-\theta_8)(1-\theta_7)R|T|$ good
(candidate, repetition) pairs. This implies that at least $(1-\theta_1)T$ candidates are isolated and land in good buckets in at least $(1-\theta_2)R$ repetitions, provided that $1-\theta_1\theta_2 \leq (1-\theta_8)(1-\theta_7)$. This proves the claim.

Condition on the event above. The actual noise in bucket $b$ in repetition $r$ is
\[
W_{b,r} = \sum_{\substack{i\not\in T\\ h_r(i)=b}} \sigma_{i,r} x_i,
\]
and by Markov's inequality,
\[
\Pr\left\{W_{b,r}^2 \geq \frac{\epsilon}{16k}\|x_{T^c}\|_2^2 \right\} \leq \frac{1}{10}.
\]
Let
\[
X_{b,r} = \mathbf{1}_{W_{b,r}^2\leq \frac{\epsilon}{16k}\|x_{T^c}\|_2^2}.
\]

We proved above that with high probability at least $(1-\theta)T$ candidates are isolated and land in good buckets in at least $(1-\theta_2)R$ repetitions; let $T'$ be the set of those $(1-\theta_1)$ candidates and for each $i\in T'$ let $R_i$ be the set of repetitions in which $i$ is isolated and land in good buckets. If $X_{h_r(i),r}=1$ for at least $\frac{1}{2(1-\theta_2)}$ fraction of $r\in R_i$, then
\[
|x_i - \hat x_i|^2 \leq \median_r W_{h_r(i),r}^2 \leq \quant_{r\in R_i}\left(W_{h_r(i),r}^2, \frac{1}{2(1-\theta_2)}\right)\leq \frac{\epsilon}{16k}\|x_{T^c}\|_2^2 \leq \frac{\epsilon}{16k}\|x_{-k}\|_2^2
\]
as desired.

The claimed result would follow from a bound on the probability that there exist at most $\theta_{11} k$ elements $i\in T'$ such that $X_{h_r(i),r}=1$ for at most a $\frac{1}{2(1-\theta_2)}$ fraction of $r\in R_i$. Observe that $W_{b,r}$ are independent since the earlier conditioning has fixed the hash functions. It follows from a Chernoff bound that
\[
p := \Pr\left\{\sum_{r\in R_i} X_{h_r(i),r} \geq \frac{|R_i|}{2(1-\theta_2)}\right\} \leq \frac{\theta_6\epsilon\delta^{2/(\theta_6 k)}}{8c_T},
\]
provided that $c_R$ is large enough. In expectation there are $|T'|p \leq \theta_6 k \delta^{2/(\theta_6 k)} / 8$ elements in $T'$ with bad estimates. Another Chernoff bound over $i\in T'$ gives the overall failure probability at most $\delta$. The total number of missed candidates is at most $(\theta_1+\theta_6)
|T|$.
\end{proof}

These two parts together will give us a weak $\ell_2/\ell_2$ system. We conclude with:
\begin{theorem}\label{thm:weaksystem}
There exists a randomized construction of a matrix $M$ with $\Oh((\frac{k}{\epsilon} + \frac{1}{\epsilon}\log\frac{1}{\delta}) \log (n/k) )$ rows, which with probability $1-\delta$ is a $(k,\zeta,\epsilon)$ weak $\ell_2/\ell_2$ system. 
\end{theorem} 
\begin{proof}
The matrix $M$ is the concatenation of the matrix of a $(k,\zeta/2)$-weak identification system with the $(\Oh(k/\epsilon),\delta/2,\epsilon)$-guarantee and the estimation matrix in Lemma~\ref{lem:estimation} (where $\delta$ is replaced with $\delta/2$ and $\zeta$ with $\zeta/2$). The weak identification system, by Lemma~\ref{lem:weak_identification}, returns a set $T$ of candidate indices which misses at most $\frac{\zeta}{2}k$ elements of $H_{k,\epsilon}(x)$, with probability at least $1-\frac{\delta}{2}$. Then by Lemma~\ref{lem:estimation} the estimation process gives `bad' estimates to at most $\frac{\zeta}{2} k$ elements in $T$ with probability at least $1-\frac{\delta}{2}$. Then we truncate $\hat x$ to the largest $k$ coordinates. The claim then follows from the same argument for \cite[Lemma 4]{PS12} or \cite[Theorem 14]{GLPSl1}, with the only change as follows: if some $i\in H_{k,\epsilon}(x)$ with a good estimate is replaced with some $j\not\in H_{k,\epsilon}(x)$ with a good estimate, it then follows from the good estimate guarantee that $\mu \leq x_i\leq 5\mu/4$ and $3\mu/4\leq x_j\leq \mu$, where $\mu=\sqrt{\epsilon/k}\|x_{-k}\|_2$. Thus $|\hat x_j - x_i| \leq |\hat x_j - x_j| + |x_j - x_i| \leq \mu/4 + \mu/2 = 3\mu/4$ and there are at most $k$ such replacements, which introduces squared $\ell_2$ error into $\hat z$ of at most $k(3\mu/4)^2 \leq (9/16)\epsilon\|x_{-k}\|_2^2$. Finally, the overall success probability is at least $1-\delta$.
\end{proof}

\subsection{Overall algorithms}

In this section we show how to combine different weak $\ell_2/\ell_2$ systems with the existing algorithms presented in Section~\ref{sec:toolkit}, to get our desired algorithm.

We move on with our first theorem, which can be compared with the result in \cite{gilbert2013l2}. 

\begin{theorem}\label{thm:k^2_algorithm1}
There exists a recovery system $\mathcal{A} = (\mathcal{D},\mathcal{R})$ which satisfies the $\ell_2/\ell_2$ guarantee with parameters $\left(n,k,\epsilon,\Oh\left(\frac{k}{\epsilon} \log\frac{n}{k}\right),e^{-\frac{k}{\log^3 k}}\right)$. Moreover, $\mathcal{R}$ runs in $\Oh\left(k^2 \log^{2+\gamma}(\frac{n}{k})\right)$ time, for any constant $\gamma$.
\end{theorem} 

\begin{proof}
Let $\ell = \lceil \log_3 k\rceil + 1$ and we shall pick $\ell$ weak systems $W_1,\dots,W_\ell$. Let $C$ be a constant such that $C\log\log k = \log^3 k$. For $1 \leq i \leq C \log \log k$ we pick $W_i$ with parameters $(\frac{k}{3^i}, \frac{1}{3}, \frac{\epsilon}{2^i} )$ and target failure probability $e^{-\frac{ck}{3^i}}$. For $ C \log \log k +1 \leq i \leq \ell$ we pick $W_i$ with parameters $(\frac{k}{3^i},\frac{1}{3}, \frac{\epsilon}{\log k})$ with target failure probability $e^{-k/ \log^3 k}$.

\noindent\textbf{Number of Measurements.} For $i\leq C\log \log k$, each $W_i$ takes $\frac{k}{\epsilon}(\frac 23)^i\log\frac{3^i n}{k}$ measurements, which sum up to $\Oh\left(\frac{k}{\epsilon}\log\frac{n}{k}\right)$. For $i>C\log\log k$, each $W_i$ takes $\Oh\left(\frac{k}{\epsilon}(\frac{\log k}{3^i}+\frac{1}{\log^2 k})\log\frac{3^i n}{k}\right) = \Oh\left(\frac{k}{\epsilon\log^2 k}\log\frac{3^i n}{k}\right)$ measurements (by our choice of $C$), which sum up to $\Oh\left(\frac{k}{\epsilon}\log\frac{n}{k}\right)$. The overall number of measurements is therefore $\Oh\left(\frac{k}{\epsilon}\log\frac{n}{k}\right)$.


\noindent\textbf{Runtime.} For $i=1,2,\dots,\ell$ with $x^{(1)} = x$  we run the algorithm of the weak-system on $W_i x^{(i)}$ to find a vector $r^{(i)}$. Then we set $x^{(i+1)} \leftarrow x^{(i)} - r^{(i)}$ and observe that $W_i x^{(i+1)} = W_i(x^{(i)} - r^{(i)}) =  W_ix^{(i)} - W_i r^{(i)}$. A standard analysis as in \cite{GLPS12} gives the desired result. 
 \end{proof}

We now move with the next theorem, which is an improvement upon \cite{GLPS12} in terms of the failure probability, as well as runtime when $k/\epsilon \leq n^{1-\gamma}$.

\begin{theorem}

There exists a recovery system $\mathcal{A} = (\mathcal{D},\mathcal{R})$ which satisfies the $\ell_2/\ell_2$ guarantee with parameters $\left(n,k,\epsilon,\Oh\left(\frac{k}{\epsilon}\log\frac{n}{k}\right),e^{-\frac{\sqrt{k}}{\log^3 k}}\right)$. Moreover, $\mathcal{R}$ runs in time $\Oh\left(\frac{k}{\epsilon}\log^{2+\gamma} n\right)$ for $k/\epsilon\leq n^{1-\alpha}$ and in time $\Oh\left(n \log^2n \log k \right)$ for $k/\epsilon\geq n^{1-\alpha}$, where $\gamma , \alpha >0$ are arbitrary constants and the constants in the $\Oh$-notations depend on $\gamma,\alpha$.
\end{theorem}

\begin{proof}
Let $\ell = \lfloor C\frac{1}{2}\log k\rfloor$. We shall pick $\ell$ weak systems $W_1,\dots,W_\ell$ using Theorem \ref{thm:weaksystem}. For $1 \leq i \leq \ell$ we pick $W_i$ with parameters $(\frac{k}{3^i}, \frac{1}{3}, \frac{\epsilon}{2^i})$ and failure probability $e^{-\frac{ck}{3^i}}$. Then we invoke Theorem \ref{thm:k^2_algorithm1} for $K = \sqrt{k}$. Our randomized matrix is then the vertical concatenation of all $W_i$ with $A$. Observe that the total number of rows is $\Oh(\frac{k}{\epsilon} \log \frac{n}{k})$. 

For $i=1,\dots,\ell$ with $x^{(1)} = x$  we run the algorithm of the weak-system on $W_i x^{(i)}$ to find a vector $r^{(i)}$. Then we set $x^{(i+1)} \leftarrow x^{(i)} - r^{(i)}$ and observe that $W_i x^{(i+1)} = W_i(x^{(i)} - r^{(i)}) =  W_ix^{(i)} - W_i r^{(i)}$, which can be computed in $\Oh(k \log n)$ time. After $\ell$ iterations we run the algorithm guaranteed by Theorem \ref{thm:k^2_algorithm1} on vector $x^{(\ell+1)}$ with matrix $A$ to get candidate set $T$. We then set
\[ 
S \leftarrow T \cup \bigcup_{i=1}^{\ell} \mathrm{supp} (r^{(i)}).
\]
We observe that previous analyses, such as \cite{GLPS12}, immediately imply that 
\[\|x_S\|_2^2 \leq (1+ \epsilon) \|x_{-k}\|_2^2.\]

The failure probability is dominated by the failure probability of the the weak systems at levels $\ell$ plus the failure probability of the procedure associated with $A$, hence the desired reuslt.\end{proof}

The next theorem is an improvement on the main result in \cite{gilbert2013l2}. Our algorithm achieves the optimal dependence on $\epsilon$ and better failure probability, using Theorem~\ref{thm:gilbert_identification} for identification and Lemma~\ref{lem:estimation} for estimation. The improvement primarily comes from two changes: (i) an improved analysis of \textsc{Count-Sketch}, namely Lemma~\ref{lem:estimation}, and (ii) better choice of parameters of the weak system in each iteration, which are set identically to those in the proof of Theorem~\ref{thm:k^2_algorithm1}.

\begin{theorem}
Suppose that $k = n^{\Omega(1)}$. There exists a recovery system $\mathcal{A}= (\mathcal{D}, \mathcal{R})$ which satisfies the $\ell_2/\ell_2$ guaranteee with parameters $\left(n,k,\epsilon, \Oh(\frac{k}{\epsilon} \log\frac{n}{\epsilon k}), (\frac{n}{k})^{-\frac{k}{\log k}}\right)$. Moreover, $\mathcal{R}$ runs in $\Oh(\frac{1}{\epsilon}k^2 \poly(\log n))$ time.
\end{theorem}

\begin{proof}[Proof (Sketch)] 
The proof is the same as before but for identification instead of using a weak system we invoke Theorem \ref{thm:gilbert_identification} to find a set $S$ such that $|H_{k/\epsilon,1}(x) \setminus S|\leq (1/2)k/\epsilon$. Accurate estimates of them can be found using Lemma \ref{lem:estimation}. Repeating the same proof as before, but iterating $\log(k/\epsilon)$ times we get the desired result.
\end{proof}

\section{Adaptive Sparse Recovery}\label{sec:adaptive}






%
%
%
%

%
%
%
%
%

\subsection{$1$-sparse Adaptive Compressed Sensing}

\begin{lemma}[{\cite[Lemma 3.2]{indyk2011power}}]\label{thm:lemma32}
Let $x \in \mathbb{R}^n$ and suppose that there exists a $j$ with $|x_j| \geq C\frac{B^2}{\delta^2} \|x_{[n]\setminus\{j\}}\|_2$ for some constant $C$ and parameters $B$ and $\delta$. With two non-adaptive measurements, with probability $1- \delta$ we can find a set $S \subset [n]$ such that (i) $j \in S; (ii) \|x_{S\setminus\{j\}}\|_2 \leq \frac{1}{B}\|x_{[n]\setminus \{j\}}\|_2$ and (iii) $|S| \leq 1 + \frac{n}{B^2}$.

\end{lemma}

The authors in \cite{indyk2011power} apply the aforementioned lemma $\Oh(\log \log n)$ times with appropriate parameters and obtain an algorithm with $\Oh(\log \log n)$ measurements. However, their approach gives only constant success probability. We shall show how to boost the success probability, by first running a preconditioning algorithm and then applying their Lemma with similar parameters as they do (not exactly the same though).

We first prove the following lemma, which will serve as a preconditioning.

\begin{lemma}\label{thm:preconditioner}
Let $x\in \R^n$ and suppose that there exists a $j$ with $|x_j| \geq  5\|x_{[n]\setminus\{j\}}\|_2$. Then there exists a scheme that uses $\Oh_{b,c}(\log \log n)$ measurements and with probability $1 - \frac{1}{\log^{c} n}$ finds a set $S$ of size $\frac{n}{\log n }$ such that (i) $j \in S$ and (ii) $\|x_{S\setminus \{j \}}\|_2 \leq \frac{1}{ \log^b n} \|x_{[n]\setminus \{j\}}\|_2$, where $b,c$ are absolute constants which can be made arbitrarily large.
\end{lemma}

\begin{proof}
We assume that the coordinates of $x$ are randomly permuted because we can apply a random permutation $\pi$ to the vector $x$, then find a set $S$ satisfying the conditions of the lemma, and at the end compute $\pi^{-1}(i)$.

Let $\operatorname{enc}: \{0,1\}^{\alpha \log \log n} \rightarrow \{0,1\}^{C_0 \log \log n}$ be the encoding function of an error-correcting code $E$ that corrects a $0.45$-fraction of errors, where $\alpha$ is a constant. Such codes exist, see, e.g.,~\cite{spielman96}. Denote by $\operatorname{enc}_b(i)$ the $b$-th bit of $\operatorname{enc}(i)$. Define $\operatorname{trunc}: [n] \rightarrow \{0,1\}^{\alpha \log \log n}$ to be a function such that $\operatorname{trunc}(i)$ equals the first $ \alpha \log \log n$ bits in the binary representation of $i$. Let also $\sigma:[n] \rightarrow \{+1,-1\}$ be a $2$-wise independent hash function. Then we perform the following $2\cdot C_0 \log \log n $ measurements:
\[ 
V_{b,0} = \sum_{i:\operatorname{enc}_b(\operatorname{trunc}(i)) = 0} \sigma_i x_i, \qquad  V_{b,1} = \sum_{i:\operatorname{enc}_b(\operatorname{trunc}(i)) = 1} \sigma_i  x_i,
\]
for all $b = 1, \dots , C_0 \log \log n$.

We form a binary string $r$ of length $C_0 \log \log n$ as follows: for each $b=1,\dots,C_0\log\log n$, $r_b=1$ if $|V_{b,1}| > |V_{b,0}|$ and $r_b = 0$ otherwise. At the end we find in the error-correcting code $E$ the closest codeword to $r$, say $r'$. Define  $S$ to be the set of all $i$ such that $\operatorname{trunc}(i) = r'$.  

We now show correctness. Let $ J = \operatorname{enc}(\operatorname{trunc}(j))$ and for $q = 0,1$ let $I_b(q) = \{i: \operatorname{enc}_b(\operatorname{trunc}(i)) = q\}$. Observe that 

\[ \mathbb{E}[V^2_{b,J_b}] = x_j^2 + \sum_{i\in I_b(J_b)\setminus\{j\}} x_i^2, \]

and 

\[\mathbb{E}[V^2_{b,\overline{J_b}}] = \sum_{i\in I_b(\overline{J_b}) } x_i^2.  \]

Observe that by Markov's inequality,
\[
\Pr\left\{ \sum_{i\in I_b(J_b)\setminus\{j\}} x_i^2 \geq 5 \left\|x_{I_b(J_b)\setminus\{j\}}\right\|_2^2 \right\} \leq \frac{1}{5}
\]
and
\[
\Pr\left\{ \left|V^2_{b,\overline{J_b}}\right| \geq 5 \left\|x_{I_b(\overline{J_b})}\right\|_2^2\right\}\leq \frac{1}{5}.
\]
It follows that with probability at least $3/5$,
\[
\left|V_{b,J_b}\right| \geq |x_j| - \sqrt{5} \left\|x_{I_b(J_b)\setminus\{j\}}\right\|_2 \geq \sqrt{5} \left\|x_{I_b(\overline{J_b})}\right\|_2 \geq \left|V_{b,\overline{J_b}}\right|. 
\] 
This means that, by a Chernoff bound, the string $r$ will agree with $J_b$ at least in a $0.55$-fraction of positions with probability at least $1-\frac{1}{\log^{\gamma C_0} n}$ , where $\gamma$ is an absolute constant. Applying the decoding algorithm of $E$ we can recover $\operatorname{trunc}(j)$. Observe now that our set $S$ has $\frac{n}{\log^{\alpha} n}$ coordinates. Since we assumed that the order of elements in $x$ is random, we have that $\mathbb{E}\|x_{S\setminus\{j\}}\|_2^2 =\mathbb{E}\|x_{ i \neq j: \operatorname{trunc}(i) = \operatorname{trunc}(j) }\|_2^2 = \frac{1}{\log^{\alpha} n} \|x_{[n]\setminus \{j\}}\|_2^2 $. It follows from Markov's inequality that $\|x_{S\setminus \{j\}}\|_2^2\geq \frac{1}{\log^{a/2} n}  \|x_{[n]\setminus \{j\}}\|_2$ with probability at most $\frac{1}{\log^{\alpha/2 } n}$. Hence the overall failure probability is $\frac{1}{\log^{\gamma C_0}n}+\frac{1}{\log^{\alpha/2}n}$. This concludes the proof of the lemma.
\end{proof}

We are now ready to state and prove the main theorem of this section.

\begin{theorem} \label{thm:onesparse}
Let $x \in \mathbb{R}^n$ and suppose that there exists $j$ such that $|x_j| \geq C \|x_{-k}\|_2$. Then, there exists an algorithm \textsc{OneSparseRecovery} that performs $\Oh(\log \log n)$ measurements in $\Oh( \log \log n)$ rounds, and finds $j$ with probability $1 - \frac{1}{\log^cn}$, where $c>0$ is some absolute constant.
\end{theorem}

\begin{proof}
We first apply Lemma~\ref{thm:preconditioner} on vector $x$ and obtain a set $S$ and then follow the approach in \cite{indyk2011power} for $x_S$. Consider now the following sequence of parameters:

\[ B_0 = 2 , \quad B_i = B_{i-1}^{\frac{3}{2}}, \quad \delta = \frac{1}{\log^c n}.\]

Let $r$ be the first index such that $B_r \geq n$. For each $i = 1, \dots, r$, we apply Lemma~\ref{thm:lemma32} to $x_{S_{i-1}}$ with parameters $B=B_i$ and $\delta=\delta_i$ and obtain a set $S_i$, where $S_0$ is taken to be the set $S$. It is easy to see that $r = \Oh( \log \log n)$.  We shall inductively prove that at all steps $|x_j| \geq \frac{B_i^2}{\delta^2} \|x_{S_i}\|_2$. The base case follows immediately from Lemma~\ref{thm:preconditioner}. By the induction step, at step $i$ we find a set $S_i$ such that $ \|x_{S_i}\|_2 \leq \frac{1}{B_{i-1}} \|x_{S_{i-1}}\|_2 \leq \frac{1}{B_{i-1}} \frac{\delta_i^2}{B_{i-1}^2}|x_j|$ or, equivalently, $|x_j| \geq \frac{B_i^2}{\delta^2} \|x_{S_i}\|_2$. After $r$ iterations, $|S_r|\leq 1 + n/B_r^2 < 2$ and thus we have uniquely identified $j$.

The overall failure probability is at most $ \frac{1}{\log^c n} + r\delta \leq  \frac{1}{\log^{c'}n}$ for some $c'>0$. The overall number of measurements is $\Oh(\log \log n)$, since in every round we use $2$ measurements, plus $\Oh( \log \log n)$ measurements for the application of Lemma \ref{thm:preconditioner} at the very beginning.
\end{proof}

\subsection{$k$-sparse recovery}

\subsubsection{Results}

\begin{theorem}[Whole regime of parameters]
Let $ x \in \mathbb{R}^n$ and $\gamma>0$ be a constant. There exists an algorithm that performs $\Oh( (k/\epsilon) \log \log(\epsilon n/k))$ adaptive linear measurements on $x$ in $\Oh( \log^*k \cdot \log \log ( \epsilon n/k))$ rounds, and finds a vector $\hat{x} \in \mathbb{R}^n$ such that $ \| x - \hat{x}\|_2^2 \leq (1+\epsilon) \|x_{-k}\|_2^2$. The algorithm fails with probability at most $\exp(-k^{1-\gamma})$.
\end{theorem}

\begin{theorem}[low sparsity regime]\label{thm:ell_infty/ell_2_adaptive}
Let $ x \in \mathbb{R}^n$ and parameters $k,\epsilon$ such that $k/\epsilon \leq c \poly(\log n)$. There exists an algorithm that performs $\Oh( (k/\epsilon) \log \log n)$ adaptive linear measurements on $x$ in $\Oh( \log \log n)$ rounds, and finds a vector $\hat{x} \in \mathbb{R}^n$ such that $ \| x - \hat{x}\|_2 \leq (1+\epsilon) \|x_{-k/\epsilon} \|_2$. The algorithm fails with probability at most $1/\poly(\log n)$.
\end{theorem}

\subsubsection{Whole Regime of Parameters}



Let $\gamma\ll 1$ be an absolute constant and $T$ be a constant depending only on $\gamma$. Let a sequence $\gamma_0,\gamma_1,\ldots, \gamma_{T}$ be a decreasing arithmetic progression such that $\gamma_0 = 1-\gamma$ and $\gamma_{T} = \gamma$. Let also $C$ be an absolute constant.

The algorithm finds a constant fraction of heavy coordinates of $x$ by hashing all coordinates to a number of buckets. Then it observes their values, subtracts them from $x$ and iterates by changing $(k,\epsilon,B)$ parameters. The way these parameters change is crucial in obtaining the desired low failure probability. We split the algorithm in three different phases. Let $\phi \in \{1,2,3\}$ be the variable corresponding to the phase the algorithm is in. The overall algorithm is presented in Algorithm~\ref{alg:adaptive_overall}. In the $r$-th round of phase $\phi$, we pick for each $r\in[R_r^{(\phi)}]$ hash functions $h_{r,j}^{(\phi)}:U \rightarrow [C'k_r^{(\phi)}/\epsilon_r^{(\phi)}]$, where $C'$ is a large absolute constant and $U$ is the universe of the elements we have not found thus far; in the beggining it is $[n]$, and then it becomes smaller and smaller as we find more and more elements. We then run the $1$-sparse recovery algorithm in each bucket induced by the hash functions, as illustrated in \textsc{HashAndRecover} function below. All the random signs and the hash functions are fully independent; we do not elaborate on that, since the independence and the space is not a primary goal in the compressed sensing literature

\begin{algorithm}
\begin{algorithmic}
\Function{HashAndRecover}{$x,k,\epsilon,R,U$}
	\For{$j=1$ to $R$}
		\State Pick a hash function $h_j: U \to [C'k/\epsilon]$
		\State $S_j \gets \bigcup_{l \in [C' k/\epsilon]} \Call{OneSparseRecovery}{x_{h_j^{-1}(l)}}$
	\EndFor
	\State \Return $\bigcup_{j \in [R] } S_j$
\EndFunction
\end{algorithmic}
\end{algorithm}

\begin{algorithm}
\caption{Adaptive sparse recovery algorithm}\label{alg:adaptive_overall}
\begin{algorithmic}
\State $J \gets \emptyset$
\For{$r=0$ to $\log \log k-1$}
	\State $k_r^{(0)} \gets \frac{k}{2^r}$
	\State $\epsilon_r^{(0)} \gets \epsilon \cdot (\frac{3}{4})^r$
	\State $R_r^{(0)} \gets C$
	\State $J_{aux} \leftarrow \Call{HashAndRecover}{x_{\bar J},k_r^{(0)},\epsilon_r^{(0)},R_r^{(0)}, [n]\setminus J}$
	\State $J \gets J\cup J_{aux}$
\EndFor
\For{$r=0$ to $\log^*{k^\gamma}$}
	\State $k_r^{(2)} \gets \frac{k}{(2 \uparrow \uparrow r)\cdot\log k}$
	\State $\epsilon_r^{(1)} \leftarrow \epsilon$
	\State $R_r^{(2)} \gets C\log k$
	\State $J_{aux} \leftarrow \Call{HashAndRecover}{x_{\bar J},k_r^{(1)},\epsilon_r^{(2)},R_r^{(2)}, [n] \setminus J}$
	\State $J\gets J\cup J_{aux}$
\EndFor
\For{$r=0$ to $T$}
	\State $k_r^{(3)} \gets k^{\gamma_r}$
	\State $\epsilon_r^{(3)} \leftarrow \epsilon$
	\State $R_r^{(3)} \gets Ck^{1-\gamma_r}$
	\State $J_{aux} \leftarrow \Call{HashAndRecover}{x_{\bar J},k_r^{(3)},\epsilon_r^{(3)},R_r^{(3)}, [n] \setminus J}$
	\State $J \gets J\cup J_{aux}$
\EndFor
\State \Return $x_J$
\end{algorithmic}
\end{algorithm}

We need the following lemma.
\begin{lemma} In each phase and every round $r$, it holds with probability at least $1 - \exp(-\Omega( k^{1-\gamma}))$ that $\left|H_{k_r^{(\phi)},\epsilon_r^{(\phi)}}(x_{\bar J})\right|\leq k_{r+1}^{(\phi)}$.
\end{lemma}

\begin{proof}
For some point dring the execution of the algorithm let $ z = x_{\bar{J}}$ and define the following (bad) events.
\begin{itemize}
\item $\mathcal{B}^{(1)}_{j,r,l}$: $ | \{i \in H_{k,\epsilon}(z): h_{r,j}(i) =l \}| > 1$
\item $\mathcal{B}^{(2)}_{j,r,l}$: $ \|z_{ \{i \in [n]\setminus H_{k,\epsilon}(z): h_{r,j}(i) = l \}} \|_2^2 \geq \frac{C''\epsilon_r}{k_r} \|z_{-Ck_r}\|_2^2$
\item $\mathcal{B}_{i,r,j}$: $\mathcal{B}^{(1)}_{j,r,h_{r,j}(i)} \text{ or } \mathcal{B}^{(2)}_{j,r,h_{r,j}(i)}$
\end{itemize}

%
%

Intuitively, the event $\mathcal{B}^{(1)}_{j,r,l}$ means the $l$-th bucket associated with the $j$-th hash function in phase $r$ contains two or more `head' elements (elements in $H_{k,\epsilon}(z)$), and the event $\mathcal{B}^{(2)}_{j,r,l}$ means that the same bucket contains a lot of noise in $\ell_2$ norm. In each iteration $r$, a heavy hitter fails to be recovered if it fails to be recovered in all $R_r$ repetitions, and we claim that at most $k_{r+1}$ out of the $k_r$ heavy hitters will not be recovered with high probability.

Fix the $k_{r+1}$ heavy hitters that fail to be recovered, and condition on the hashing of all other coordinates. Consider hashing these $k_{r+1}$ heavy hitters one by one. Since the hash functions are fully independent, the resulting distribution does not depend on the order. Consider an individual repetition. Let $N_t$ denote the number of coordinates $i$ among the $k_{r+1}$ heavy hitters such that $\mathcal{B}_{i,r,j}$ happens. It is clear that $N_t$ is increasing and a standard calculation shows that with probability $p \leq \epsilon/C'$ we have $N_{t+1} = N_t + 1$. Therefore if we define $M_t = N_t - tp$, then $M_0,\dots,M_{k_{r+1}}$ is a supermartingale.

It follows from Azuma-Hoeffding inequality that 
\[
\Pr(N_{k_{r+1}} = k_{r+1}) = \Pr(M_{k_{r+1}} = k_{r+1} (1-p)) \leq \exp\left(-2(1-p)^2k_{r+1}\right).
\]
and thus
\[
\Pr(N_{k_{r+1}} = k_{r+1}\text{ in all }R_r\text{ repetitions}) \leq \exp\left(-2(1-p)^2k_{r+1}R_r\right).
\]
Taking a union bound over all possible $k_{r+1}$ heavy hitters,
\begin{multline*}
\Pr(\exists S \subseteq H_{k_r,\epsilon_r}(x_{\bar J}), |S| = k_{r+1}: \mathcal{B}_{i,r,j}\text{ happens for all }i\in S, j\in R_r)\\
\leq \binom{k_r}{k_{r+1}}\exp\left(-2(1-p)^2k_{r+1}R_r\right).
\end{multline*}

It suffices to upper bound the tail bound on the right-hand side above in each phase.

\begin{itemize}
	\item \textbf{Phase 1} ($0 \leq r \leq \log\log k - 1$):
\[
		\binom{k_r}{k_{r+1}}\exp\left(-2(1-p)^2k_{r+1}R_r\right) \leq \exp\left(-c_1\frac{k}{2^r}\right)\leq \exp\left(-c_1\frac{k}{\log  k}\right).
\]
	\item \textbf{Phase 2} ({$0 \leq r \leq \log^\ast k^\gamma$}): 
\[
	\binom{k_r}{k_{r+1}}\exp\left(-2(1-p)^2k_{r+1}R_r\right) \leq \exp(-c_2 k_{r+1}\log k)\leq \exp(-c_2' k^{1-\gamma}).
\]
	\item \textbf{Phase 3} ({$0 \leq r \leq T$}): When $r<T$,
\[
	\binom{k_r}{k_{r+1}}\exp\left(-2(1-p)^2k_{r+1}R_r\right) \leq \exp(-c_2 k_{r+1}R_r)\leq \exp(-c_3'k^{1-(1-2\gamma)/T}) \leq \exp(-c_3'k^{1-\gamma})
\]
provided that $(T+2)\gamma \geq 1$. 
In the last step, i.e.\@ $r=T$, for a coordinate $i$,
\[
\Pr\left\{\mathcal{B}_{i,r,j} \text{ happens for all } j \in [R_r]\right\} \leq \exp(-c_4k^{1-\gamma}). 
\]
This allows us to take a union bound over all $i \in H_{k_T,\epsilon}(x_{\bar{J}})$, so in this step we shall recover all of them with probability $1 - \exp(-c_4'k^{1-\gamma})$.
\end{itemize}
The proof of the lemma is complete.
\end{proof}

\subsubsection{Low-Sparsity Regime}

We need the following lemma, which is standard in the sparse recovery and streaming algorithms literature.





\begin{lemma}\label{lem:countsketch}
Let $x \in \mathbb{R}^n$ and $\mathcal{F} = \{F_1,\ldots,F_U \}$ be a partition of $[n]$. For a set $S \subseteq [n]$ define $F(S) = \bigcup_{i\in F_j} F_j$. There exists an algorithm that performs $\Oh((k/\epsilon)\log |U|)$ non-adaptive measurements and with probability $1- |U|^{-c}$ finds a set $T \subseteq [U]$ of size $\Oh(k)$ such that 
$F(H_{k,\epsilon}(x)) \subseteq T$. 
\end{lemma}

\begin{proof}[Proof (Sketch)]
For each $r=1,\ldots, \log |U|$ we pick a $2$-wise independent hash function $h_r:[U] \rightarrow [Ck]$, for some absolute constant $C$ and a $2$-wise independent hash function $\sigma_r:[n] \rightarrow \{+1,-1\}$. Then, for every $r \in [\log |U|]$ and $j \in [Ck]$ we perform  measurement

\[ y_{j,r} = \sum_{t \in [U] : h_r(t) =j} \sum_{i \in F_t} \sigma_{i,r} x_i. \]

For every $j \in U$ we compute \[\hat{z}_j = \median_{1\leq r\leq \log k} \left|y_{h_r(j),r}\right| \] and find the largest $C_0 k$ indices with the biggest $\hat{z}$ values in magnitude, forming a set $T$. We then output $T$.
\end{proof}

We are now ready to prove our main result in the low-sparsity regime, Theorem~\ref{thm:ell_infty/ell_2_adaptive}.

\begin{proof}[Proof of Theorem~\ref{thm:ell_infty/ell_2_adaptive}]
Pick a hash function $h:[n] \rightarrow [\log^{c_0}n]$, where $c_0$ is a large absolute 
constant. Observe that with probability $1 - \frac{1}{\poly(\log n)}$ every $h^{-1}(p)$ 
contains at most $1$ element of $H_{k,\epsilon}(x)$. Let $\mathcal{F}_{\operatorname{good}} = \{h^{-1}(p): |h^{-1}(p) \cap H_{k,\epsilon}(x) | = 1 \}$. Then we 
can invoke Lemma~\ref{lem:countsketch} and obtain a set $S$ of size $\Oh(k)$ which with 
probability $1 - \frac{1}{\poly(\log n)}$ is a superset of $\mathcal{F}_{\operatorname{good}}$. The number of measurements we need is $\Oh(\frac{1}{\epsilon}k \cdot \log \log n)$. We then run the $1$-sparse recovery algorithm guaranteed by Theorem~\ref{thm:onesparse} in each $F_p \in \mathcal{F}_{\operatorname{good}}$, for a total of
$\Oh( \frac{1}{\epsilon} k \cdot \log \log n)$ measurements in $\Oh(\log \log n)$ rounds. 
Since every $1$-sparse recovery routine succeeds independently with probability $1 - \frac{1}{\poly(\log n)}$ we get that all of them succeed with the desired probability and 
hence we can obtain a set $S$ containing all elements of $H_{k,\epsilon}(x)$. By 
observing them directly in another round, we can obtain their values, and trivially satisfy the $\ell_2/\ell_2$ guarantee.
\end{proof}

\section{Spiked-Covariance Model}\label{sec:spikedcov}

In the spiked covariance model, the signal $x$ is subject to the following distribution: we choose $k$ coordinates uniformly at random, say, $i_1,\dots,i_k$. First construct a vector $y\in \R^n$, in which each $y_{k_i}$ is a uniform Bernoulli variable on $\{-\sqrt{\epsilon/k}, +\sqrt{\epsilon/k}\}$ and these $k$ coordinate values are independent of each other. Then let $z\sim N(0,\frac{1}{n}I_n)$ and set $x = y+z$. We now present a non-adaptive algorithm (although the runtime is slow) that uses $\Oh((k/\epsilon)\log(\epsilon n/k) + (1/\epsilon)\log(1/\delta))$ measurements and a matching lower bound.

\begin{theorem}
Assume that $(k/\epsilon)\log(1/\delta) \leq \beta n$, where $\beta\in(0,1)$ is a constant. There exists an $\ell_2/\ell_2$ algorithm for the spiked-covariance model that uses $\Oh\left(\frac{k}{\epsilon}\log\frac{\epsilon n}{k}+ \frac{1}{\epsilon}\log\frac{1}{\delta}\right)$ measurements and succeeds with probability $\geq 1-\delta$. Here the randomness is over both the signal and the algorithm.
\end{theorem}
\begin{proof}
First note that the maximum magnitude in $z$ is $\Oh\left(\sqrt{\frac{\log(1/\delta)}{n}}\right)$ with probability $\geq 1-\delta$, which is smaller than $\sqrt{\epsilon/k}$ given the assumption on $k$. Condition on this event. Furthermore, with probability $\geq 1-e^{-\Omega(n)}$, it holds that $(1-\eta)\frac{n-k}{n}\leq \|x_{-k}\|_2^2\leq (1+\eta)\frac{n-k}{n}$. Further condition on this event.

We essentially repeat Lemma~\ref{lem:estimation} with $T = [n]$. That is, we estimate each coordinate, discard the estimates outside the range $[(1-\gamma)\sqrt{\frac{\epsilon}{k}},(1+\gamma)\sqrt{\frac{\epsilon}{k}}]$, where $\gamma = \alpha\sqrt{(1+\eta)(1-\frac kn)}$, and retain only the top $k$ coordinates. To accommodate the larger size of $T$, the number of repetitions needs to be $R=c_R(\log\frac{\epsilon n}{k} + \frac{1}{k}\log\frac{1}{\delta})$, so the total number of measurements is $\Oh\left(\frac{k}{\epsilon}\log\frac{\epsilon n}{k}+ \frac{1}{\epsilon}\log\frac{1}{\delta}\right)$ as claimed. 

Next we show correctness. With probability $\geq 1-\delta$, the top $k$ items we retain contain at least $(1-\theta)k$ elements of $H_k(x)$, and the estimation error of each of them is at most $\beta\sqrt{\epsilon/k}\|x_{-k}\|_2$. Hence these elements survive the thresholding. In total these well-recovered heavy hitters contribute at most $\beta^2\epsilon \|x_{-k}\|_2^2$ to the approximation error. For the remaining heavy hitters, they could be (i) unrecovered, (ii) replaced by spurious ones with magnitude bounded by $(1+\gamma)\sqrt{\epsilon/k}$, or (iii) have estimation error at most $\gamma\sqrt{\epsilon/k}$. Hence they contribute to the residual energy (squared $\ell_2$ norm) at most $(1+\gamma)^2\theta\epsilon$ in total. Therefore we conclude that
\[
\|\hat x - x\|_2^2 \leq \left(1 + (\beta^2+(1+\gamma)^2\theta)\epsilon\right)\|x_{-k}\|_2^2.\qedhere
\]
\end{proof}

\begin{remark}
The algorithm above runs in time $\tilde \Oh(n)$. Alternatively, when $k/\epsilon=n^{\Omega(1)}$, we can invoke Theorem~\ref{thm:gilbert_identification} to identify a constant fraction of the heavy hitters and then estimate them as before. Overall it uses $\Oh\left(\frac{k}{\epsilon}\log\frac{n}{k} + \frac{1}{\epsilon}\log\frac{1}{\delta}\right)$ measurements and succeeds with probability at least $1-\delta$. The error guarantee follows similarly as in the proof above. This alternative algorithm uses the optimal number of measurements when $\epsilon$ is a constant (and only slightly more measurements in the general case) but runs significantly faster in time $\Oh(k^{1+\alpha}\poly(\log n))$.
\end{remark}

We prove a matching lower bound to conclude this section.
\begin{theorem}
    Suppose that $\delta < \delta_0$ for a sufficiently small absolute constant $\delta_0\in (0,1)$ and $n \geq  \left \lceil 64 \epsilon^{-1} \log(6/\delta) \right \rceil$. Then any $\ell_2/\ell_2$-algorithm that solves with probability $\geq 1-\delta$ the $\ell_2/\ell_2$ problem in the spiked-covariance model must use $\Omega(\epsilon^{-1} \log(1/\delta))$ measurements. 
\end{theorem}
\begin{proof}
The lower bound proved in Theorem~\ref{thm:ell2_hhlb} corresponds to the case of $k = 1$ in the spiked covariance model. We can follow exactly the same notation and proof for the case of $k = 1$ for general $k$. Doing so, we arrive at the point that we need to bound the total variation distance between $N(0, I_r)$, and $N(0, I_r) + \sum_{i \in \supp(y)} (3 \sqrt{\epsilon n/k}) S_i \sigma_i$, where $\sigma_i$ are random signs. Note that $\left\|\sum_{i \in \supp(y)} S_i \sigma_i\right\|_2 \leq \sqrt{10(k/n)\|S\|_F^2} = \sqrt{10kr/n}$ with constant probability, and thus one must distinguish $N(0, I_r)$ and $N(0, I_r) + O(\sqrt{\epsilon r}) v$ for a unit vector $v$ with probability $1-\Theta(\delta)$. Now we can rotate $v$ to the first standard unit vector $e_1$ as before, to conclude that the same $\Omega((1/\epsilon)\log(1/\delta))$ lower bound continues to hold for $k > 1$. 
\end{proof}

\end{document}